\documentclass{amsart}

\usepackage{amssymb,mathtools,hyperref,url,enumitem,complexity,tikz,caption}
\usepackage[linesnumbered,ruled,noend]{algorithm2e}
\hypersetup{colorlinks,linkcolor={blue},citecolor={blue},urlcolor={blue}}
\usepackage[style=alphabetic,maxbibnames=99]{biblatex}

\renewbibmacro{in:}{}
\DeclareFieldFormat*{title}{#1}
\newcommand{\commentR}[1]{\tcp*[r]{\parbox[t]{4.9cm}{\raggedright\normalfont{#1}}}}
\newcommand{\commentF}[1]{\tcp*[f]{\parbox[t]{4.9cm}{\raggedright\normalfont{#1}}}}
\SetKwComment{tcp}{}{}
\let\oldnl\nl% Store \nl in \oldnl
\newcommand{\nonl}{\renewcommand{\nl}{\let\nl\oldnl}}% Remove line number for one line
\makeatother

\usepackage{etoolbox} % For \patchcmd
\makeatletter
\patchcmd{\algocf@makecaption@ruled}{\hsize}{\textwidth}{}{} % Caption to stretch full text width
\patchcmd{\@algocf@start}{-1.5em}{0em}{}{}% For // to right margin
\makeatother

\SetAlCapHSkip{0pt} % Reset left skip of caption
\newtheorem{theorem}{Theorem}[section]

\newtheorem{lemma}[theorem]{Lemma}
\newtheorem{proposition}[theorem]{Proposition}
\newtheorem{corollary}[theorem]{Corollary}
\newtheorem{conjecture}[theorem]{Conjecture}

\theoremstyle{definition}
\newtheorem{definition}[theorem]{Definition}
\newtheorem{remark}[theorem]{Remark}

\numberwithin{equation}{section}

\newcommand{\bR}{\mathbb{R}}
\newcommand{\bQ}{\mathbb{Q}}
\newcommand{\bZ}{\mathbb{Z}}
\newcommand{\bF}{\mathbb{F}}
\newcommand{\bC}{\mathbb{C}}
\newcommand{\cO}{\mathcal{O}}

\newcommand{\ba}{\mathbf{a}}
\newcommand{\bb}{\mathbf{b}}
\newcommand{\bu}{\mathbf{u}}
\newcommand{\bv}{\mathbf{v}}

\newcommand{\bx}{\mathbf{x}}
\newcommand{\fb}{\mathfrak{b}}
\newcommand{\fp}{\mathfrak{p}}
\newcommand{\Mat}{\textup{Mat}}
\newcommand{\bc}{\mathbf{c}}
\newcommand{\ff}{\mathfrak{f}}
\newcommand{\fa}{\mathfrak{a}}
\newcommand{\SVP}{\ComplexityFont{SVP}}
\newcommand{\LWE}{\ComplexityFont{LWE}}
\newcommand{\SIVP}{\ComplexityFont{SIVP}}
\newcommand{\uSVP}{\ComplexityFont{uSVP}}

\allowdisplaybreaks

\begin{document}

\title{NP-Hardness of Ideal Lattice Problems}

\thanks{This material is based upon work supported by the National Science Foundation under Grant Nos. 2336000 and 2602045.}

\author{Daniel E. Martin}
\email{dem6@clemson.edu}
\address{Clemson University, O-110 Martin Hall, 2020 Parkway Drive, Clemson, SC}



\date{\today}

\begin{abstract}We establish the worst-case hardness of several ideal lattice problems (including $\SVP$ and $\CVP$) in the $\ell_2$ norm by providing a dimension-preserving, deterministic polynomial time reduction from their generic lattice versions. The reduction constructs an ideal lattice in the canonical embedding of a number field that approximates some input lattice up to scaling and orthogonal transformation. The integers defining the ideal and the ambient number ring, in particular its discriminant, are all polynomial in bit length relative to the generic input lattice. Furthermore, the ideal is invertible, the ring is monogenic, and the number field is totally real. If the number ring is also required to be a full ring of integers, the reduction conjecturally succeeds in bounded-error quantum polynomial time.\end{abstract}

\maketitle

\section{Introduction}\label{sec:1}

The theoretical complexity of ideal-lattice reduction has come under scrutiny due to its cryptographic implications. In this paper, we prove the $\NP$-hardness of two such problems, called $\SVP$ and $\CVP$, under appropriate parameters.

Beginning with the work of Ajtai--Dwork \cite{dwork} and Hoffstein--Pipher--Silverman \cite{hoffstein}, lattices have been a significant source of hard problems in cryptography. Lattice-based encryption is unique in its benefits---presumed quantum security, full homomorphicity \cite{gentry,brakerski}, average-case hardness \cite{ajtai,regev,peikert}---but natively suffers from large key sizes. As a remedy, Hoffstein thought to work over the cyclic rings $\bZ[x]/(x^n-1)$, where the public key becomes a single element defined by only $n$ integers rather than $n^2$ (the number of integers that define a generic lattice). He developed this idea alongside Pipher and Silverman into the NTRU cryptosystem. After NTRU, the theory of cyclotomic ring-based encryption was advanced notably by Micciancio \cite{micc2} and Peikert and Rosen \cite{rosen}. The notion was generalized by Micciancio and Lyubashevsky to the rings $\bZ[x]/(f(x))$ for monic irreducible polynomials $f(x)$ \cite{lyub,micc2}. Further development by Lyubashevsky, Peikert, Regev, Stephens-Davidowitz, Rosca, Stehl{\'e}, Wallet, and others \cite{lyub2,peik,rosca} brought theoretical cryptography to its modern state: lattices are drawn from the canonical embedding of ideals in any ring of integers, or more generally, from any number ring provided the ideal is invertible. These lattices underlie the theoretical hardness of problems like Ring $\LWE$ \cite{regev,lyub2}, Polynomial $\LWE$ \cite{stehle}, and Order $\LWE$ \cite{rosca,bolb}.

While the algebraic structure of ideal lattices provides efficiency gains, it may also open the door to new attacks. Despite hints of vulnerability \cite{bauch,campbell,cramer,mildly,pellet}, cryptosystems that derive security from the worst-case hardness of ideal lattice problems are already implemented. This hardness is poorly understood compared to generic lattices, where, for example, the Shortest and Closest Vector Problems ($\SVP$ and $\CVP$) are long known to be \NP-hard \cite{ajtai2,vanEmde}. We address this disparity with a generic-to-ideal lattice reduction algorithm. 

\begin{theorem}\label{thm:main}In deterministic polynomial time, Algorithm~\ref{alg:1} approximates (to any precision) a scaled, orthogonal transformation of its input lattice with an ideal lattice of the same dimension. The ideal, defined in polynomial bit length, is coprime to the conductor (and thus invertible) in a monogenic, totally real number ring.\end{theorem} 

\begin{remark}\label{rem:random}There are two choices of subroutine for Algorithm~\ref{alg:1}: Subroutine~\ref{sub:2} is simpler but randomized, while Subroutine~\ref{sub:3} is deterministic and justifies Theorem~\ref{thm:main}. The randomized algorithm repeats until the output ring is a full ring of integers. Each iteration conjecturally succeeds with probability roughly 0.323.\end{remark}

The informal moral of Theorem \ref{thm:main} is: \emph{The general tools of algebraic number theory do not make ideal lattices inherently more vulnerable than generic lattices.} Conductor coprimality is essential to this claim. Without it, the general tools of algebraic number theory, like unique factorization or the two-generator property (see Section~\ref{ss:numbertheory}), do not apply. 

Theorem \ref{thm:main} does \emph{not} imply that tools specific to certain fields present no new vulnerability. This is because Algorithm~\ref{alg:1} does not allow absolute control over the number field in which the output ideal lives. For example, it cannot target a specific field with a preprocessed unit group. There is, however, some control over the field: The algorithm produces a matrix $A\in\Mat_n(\bZ)$ whose characteristic polynomial, call it $f(x)$, defines the number field and ring. Alternative characteristic polynomials that produce different fields can be obtained by slightly perturbing $A$ or by replacing it with $BAB^T$ for any integer matrix $B\in A\text{GL}_n(\bZ)A^{-1}$. The resulting ideal lattice still approximates the input up to orthogonal transformation and scaling. The author has not investigated whether these freedoms can produce additional properties like Galois or a near-orthonormal $\bZ$-basis for the ring. Thus the exact degree to which our approach can close the gap between generic lattices and those loved by cryptography is a subject for future research.

As a consequence of Theorem~\ref{thm:main}, hardness results for generic lattice reduction problems in the $\ell_2$ norm now apply in the ideal setting. The approximation factors below were proved for generic lattices in the preprints of Wan \cite{wan} (building on Bennett and Peikert \cite{bennett}) and Song \cite{song}. (These factors have steadily improved over the past few decades; see \cite{ajtai2,cai,miccSVP,khot,khot2,haviv} for progress regarding $\SVP_2$ assuming $\RP\neq\NP$, and see \cite{vanEmde,dumer,dinur,openai} for $\CVP_2$.)  

\begin{corollary}\label{cor:main}Restricted to the canonical embedding of ideals coprime to the conductor in totally real, monogenic number rings, $\SVP_2$ and $\CVP_2$ are \NP-hard to approximate within factors of $\sqrt{2}$ and $\smash{n^{\frac{1}{2}-\varepsilon}}$, respectively.\end{corollary}

Related results have come from Liu, Feng, and Pan. In their recent preprints \cite{liuSVP} and \cite{liuCVP}, they prove that exact $\SVP_2$ is $\NP$-hard in lattices of rank-$2$ modules over prime cyclotomic rings of integers (ideals are rank-$1$ modules), and that exact $\CVP_2$ is $\NP$-hard in lattices of principal ideals in power-of-2 cyclotomic rings of integers. 

After we fix notation and definitions in Section~\ref{sec:2}, we present an intuitive overview of Algorithm~\ref{alg:1} and prove its basic properties in Section~\ref{sec:3}. The main technical hurdle---producing invertible ideals---is avoided until Section~\ref{sec:4}, which contains our randomized solution. It is expected to produce ideals in a ring of integers, but this experimentally observed high probability of success is essentially equivalent to an unresolved conjecture of Wang and Yu \cite{wang}. Our deterministic solution is presented in Section~\ref{sec:5}. This allows us to prove Theorem~\ref{thm:main} and Corollary~\ref{cor:main}. 

\section{Background}\label{sec:2}

\subsection{Lattice problems}We state the search versions of several problems that are amenable to Algorithm~\ref{alg:1}. The algorithm could just as easily be used to reduce decision versions from the generic to the ideal setting.

For a lattice $\Lambda$, recall that $\lambda_i(\Lambda)$ is the minimal radius of a closed ball (in some fixed $\ell_p$ norm) around the origin that contains $i$ independent vectors in $\Lambda$. We use $|\bv|_p$ to denote the $\ell_p$ norm of $\bv\in\bR^n$.

\begin{definition}\label{def:svp}Given $\gamma\in[1,\infty)$ and a basis for a lattice $\Lambda$, the \emph{approximate Shortest Vector Problem}, denoted $\SVP_p$ or $\gamma$-$\SVP_p$, asks for a nonzero lattice vector of length at most $\gamma\lambda_1(\Lambda)$.\end{definition}

\begin{definition}\label{def:sivp}Given $\gamma\in[1,\infty)$ and a basis for an $n$-dimensional lattice $\Lambda$, the \emph{approximate Shortest Independent Vectors Problem}, denoted $\SIVP_p$ or $\gamma$-$\SIVP_p$, asks for $n$ independent lattice vectors of length at most $\gamma\lambda_n(\Lambda)$.\end{definition}

\begin{definition}\label{def:usvp}Given $\gamma\in[1,\infty)$ and a basis for a lattice $\Lambda$ with $\gamma\lambda_1(\Lambda)\leq\lambda_2(\Lambda)$, the \emph{unique Shortest Vector Problem}, denoted $\uSVP_p$ or $\gamma$-$\uSVP_p$, is $1$-$\SVP_p$ in $\Lambda$.\end{definition}

\begin{definition}\label{def:cvp}Given $\gamma\in[1,\infty)$, a basis for a lattice $\Lambda\subset\bR^n$, and a target $\bx\in\bR^n$, the \emph{approximate Closest Vector Problem}, denoted $\CVP_p$ or $\gamma$-$\CVP_p$, asks for $\bv\in\Lambda$ with $\displaystyle|\bv-\bx|_p\leq\gamma\min_{\bu\in\Lambda}|\bu-\bx|_p$.\end{definition}

See \cite{goldwasser} for a comprehensive overview of these problems and \cite{decade} for a focus on cryptographic applications. In \cite{bennett2}, Bennett provides a complete survey of complexity results on $\SVP$ prior to 2023. Notable subsequent progress has been made by Hair and Sahai \cite{hair} and Wan \cite{wan}.   

\subsection{Number theory}\label{ss:numbertheory}Let $\sigma_1,\dots,\sigma_r:K\hookrightarrow\bR$ be the real embeddings of a number field $K$, and let $\tau_1,\overline{\tau}_1,\dots,\tau_s,\overline{\tau}_s:K\hookrightarrow \bC$ be the complex conjugate pairs of non-real embeddings. 

\begin{definition}The \emph{canonical} (or \emph{Minkowski}) \emph{embedding} of $K$ is $\Sigma:K\hookrightarrow \bR^r\times\bC^s$ defined by $\Sigma(\alpha)=(\sigma_1(\alpha),\dots,\tau_s(\alpha))$.\end{definition} 

In contemporary literature, ``ideal lattice'' typically refers to the canonical embedding of an ideal. We maintain this convention. (Prior to \cite{lyub2}, fixing a $\bZ$-basis for $\cO$---typically $1,\dots, x^{n-1}$ when $\cO\simeq\bZ[x]/(f(x))$---and using the coefficient embedding was more popular.) Such a lattice has full rank $n\coloneqq r+2s$.

We call $K$ \emph{totally real} if $s=0$. Ideal lattices in totally real fields may be more vulnerable to attacks that leverage unit approximations because they maximize the unit group rank $r+s-1$ in degree $n$.

\begin{definition}The \emph{ring of integers} (or \emph{maximal order}), denoted $\cO_K$, consists of elements in $K$ that are roots of monic polynomials in $\bZ[x]$. An \emph{order} is a subring of $\cO_K$ that is not contained in a proper subfield of $K$. An order is \emph{monogenic} if it takes the form $\bZ[\lambda]$ for some $\lambda\in\cO_K$.\end{definition}

There are two properties of ideals $\fa\subseteq\cO_K$ that are especially useful in cryptography: they require at most two generators, meaning $\fa=\alpha\cO_K+\beta\cO_K$, reducing the $n^2$ integers needed to define a generic lattice in $n$ dimensions to only $2n$; and when $\cO_K$ is monogenic, its elements can be multiplied in $O(n\log n)$ operations with the Fast Fourier Transform. These properties also hold for an ideal $\fa$ in a non-maximal order $\cO$ provided $\fa$ is \emph{invertible}, meaning $\fa\fb$ is principal for some ideal $\fb\subseteq\cO$. We will target the stronger condition that $\fa$ is coprime to the conductor of $\cO$.

\begin{definition}The \emph{conductor} of an order $\cO$ in $K$ is the ideal \[\ff=\ff(\cO)\coloneqq\{\alpha\in\cO\mid \alpha\cO_K\subseteq\cO\}.\]\end{definition} 

Ideals in $\cO$ coprime to $\ff$ retain many other desirable number theoretic properties. For example, they represent elements of the class group (or Picard group), factor uniquely into products of primes, and are preserved by a rank $r+s-1$ unit group. (Note that $\ff(\cO_K)=\cO_K$, so all ideals in $\cO_K$ are coprime to the conductor and invertible.)

Invertibility is an indispensable hypotheses in hardness proofs for ideal variants of Regev's Learning with Errors problem \cite{lyub2,rosca,bolb,pepin}. If we ignored invertibility, we could easily reduce any generic $n$-dimensional lattice problem to the ideal setting in our favorite degree-$n$ number field, whatever that may be. Indeed, if $[K:\bQ]=n$, a scaled copy of the input lattice can be approximated by some sublattice $\Lambda\subseteq\Sigma(\cO_K)$. Since $\Lambda$ is an ideal in the order $\bZ+\text{cov}(\Lambda)\cO_K$, we're done! Unfortunately, such a reduction is useless because the ideal is almost never invertible. It neither supports the belief that ring variants of $\LWE$ are hard (because the underlying ideals in every $\SVP$-to-$\LWE$-style security proof are invertible), nor does it make the aforementioned properties of invertible ideals available to generic lattice reduction algorithms. 

Our strategy in the next section is completely different. In contrast to the appoach above, it often produces ideals in rings of integers. Even when it doesn't, the ideal can always be made invertible.

\section{Approximating with ideal lattices}\label{sec:3}

\subsection{Intuition and an example}Given an $n\times n$ integer matrix $M$, the goal is to approximate $M\bZ^n$ with a lattice $\Sigma(\fa)$ for some invertible ideal $\fa$ in an order $\cO$. 

We start by finding a symmetric matrix $A$ representing the same lattice $M\bZ^n$. Any symmetric basis will do. The simplest way to find one is perhaps computing the Smith Normal Form $S=U^{-1}MV$, where $S$ is diagonal and $U,V\in\text{GL}_n(\bZ)$. Since \[USU^{T}\bZ^n = (USV^{-1})(VU^T\bZ^n) = M\bZ^n,\] the matrix $A\coloneqq USU^T$ provides a symmetric basis. As a minimal working example, \[M=\begin{bmatrix}1 & 1 \\ 1 & -1\end{bmatrix}=\underbrace{\begin{bmatrix}1 & 0 \\ 3 & -1\end{bmatrix}\begin{bmatrix}1 & 0 \\ 0 & 2\end{bmatrix}\begin{bmatrix}1 & 1 \\ 1 & 2\end{bmatrix}}_{USV^{-1}}\] yields the symmetric matrix \[A=\begin{bmatrix}1 & 0 \\ 3 & -1\end{bmatrix}\begin{bmatrix}1 & 0 \\ 0 & 2\end{bmatrix}\begin{bmatrix}1 & 3 \\ 0 & -1\end{bmatrix}=\begin{bmatrix}1 & 3 \\ 3 & 11\end{bmatrix}.\]

An eigenvalue of $A$, say $\lambda\in\overline{\bQ}$, has a corresponding eigenvector $\bv\in\bZ[\lambda]^n$, found by taking a column of the adjugate of $M-\lambda I_n$. In the example, $\lambda=6+\smash{\sqrt{34}}$ and $\bv=(3, 5+\smash{\sqrt{34}})^T$. 

Assuming the characteristic polynomial of $A$ is irreducible, its remaining eigenvalues and eigenvectors are the algebraic conjugates of $\lambda$ and $\bv$. In particular, if $D$ is the diagonal matrix with diagonal entries $\Sigma(\lambda)$ and $E$ is the matrix of eigenvectors satisfying $AE=ED$, then the columns of $E^T$ are precisely the canonical embeddings of the entries of $\bv$, as in \[E^T=\begin{bmatrix}3 & 5+\sqrt{34} \\ 3 & 5-\sqrt{34}\end{bmatrix}.\] Furthermore, the $\bZ$-module generated by the entries of $\bv$, call it $\fa$, is an ideal in $\bZ[\lambda]$. This equates to saying that each entry of $\lambda\bv$ is a $\bZ$-linear combination of the entries of $\bv$, which is verified by comparing columns on each side of $AE=ED$. In our example, $\fa=3\bZ+(5+\smash{\sqrt{34}})\bZ$ is a nonprincipal prime ideal of norm 3 in $\bZ[\smash{\sqrt{34}}]$.

By the spectral theorem for symmetric matrices, the normalized eigenvectors of $A$ form an orthonormal basis for $\bR^n$. In other words, if $C$ is the diagonal matrix whose $i^\text{th}$ diagonal entry is the reciprocal of the length of the $i^\text{th}$ column of $E$, then $EC$ is an orthogonal matrix. In the example, \[C=\begin{bmatrix}(68+10\sqrt{34})^{-\frac{1}{2}} & 0 \\ 0 & (68-10\sqrt{34})^{-\frac{1}{2}}\end{bmatrix}.\]

The final step scales the vector formed from diagonal entries of $CD$, call it $\bc$, by a sufficiently large real number $r$ and approximates the result with a lattice point $\Sigma(\beta) \in \Sigma(\bZ[\lambda])$. Crucially, $\Sigma(\beta)$ need not be the \emph{closest} lattice point to $r\bc$; that requirement would destroy our algorithm's running time. Provided $r$ dwarfs the entries of $\bc$ in magnitude, the crudest of approximations will do (as in lines 7-12 in Algorithm~\ref{alg:1}).

\begin{figure}[htbp]
    \centering
    % Subfigure 1: The (2,0), (1,1) Lattice
    \begin{tikzpicture}[scale=0.55]
        % Clip box expanded slightly so outer labels don't get cut off
        \clip (-3.5, -3.5) rectangle (3.5, 3.5);
        
        % Draw lattice points
        \foreach \x in {-5,...,5} {
            \foreach \y in {-5,...,5} {
                \fill[black] ({\x*2 + \y*1}, {\x*0 + \y*1}) circle (0.05);
            }
        }
        
        % Draw axes (x and y labels removed)
        \draw[->, thick] (-3.5, 0) -- (3.5, 0);
        \draw[->, thick] (0, -3.5) -- (0, 3.5);
        
        % Unlabeled tick marks at 1, 2, -1, -2
        \foreach \i in {-2, -1, 1, 2} {
            \draw (\i, 0.1) -- (\i, -0.1);
            \draw (0.1, \i) -- (-0.1, \i);
        }

        \draw[red, thick, dashed] (0,0) -- (1,-1) -- (2,0) -- (1,1) -- cycle;
        
        % Labeled tick marks at 3, -3
        \foreach \i in {-3, 3} {
            \draw (\i, 0.1) -- (\i, -0.1) node[below] {\scriptsize $\i$};
            \draw (0.1, \i) -- (-0.1, \i) node[left] {\scriptsize $\i$};
        }
    \end{tikzpicture}\hfill
    % Subfigure 2: The New Middle Lattice
    \begin{tikzpicture}[scale=0.55]
        \clip (-3.5, -3.5) rectangle (3.5, 3.5);
        
        % Scale factor = 0.02
        % w1 = (36, 36), w2 = (-22.9857114, 46.9857114)
        \foreach \x in {-5,...,5} {
            \foreach \y in {-5,...,5} {
                \fill[black] ({\x*0.05*15 - \y*0.05*11.4928556845}, {\x*0.05*15 + \y*0.05*23.4928556845}) circle (0.05);
            }
        }
        
        \draw[->, thick] (-3.5, 0) -- (3.5, 0);
        \draw[->, thick] (0, -3.5) -- (0, 3.5);

        \draw[red, thick, dashed] (0,0) -- ({-0.05*11.4928556845},{0.05*23.4928556845}) -- ({0.05*(-11.4928556845+15)},{0.05*(23.4928556845+15)}) -- ({0.05*15},{0.05*15}) -- cycle;
        
        % Unlabeled tick marks (50 and 100 mapped via scale factor 0.02 to 1 and 2)
        \foreach \pos in {-2, -1, 1, 2} {
            \draw (\pos, 0.1) -- (\pos, -0.1);
            \draw (0.1, \pos) -- (-0.1, \pos);
        }
        
        % Labeled tick marks at 150, -150 (150 * 0.02 = 3)
        \foreach \pos/\label in {-3/-60, 3/60} {
            \draw (\pos, 0.1) -- (\pos, -0.1) node[below] {\scriptsize $\label$};
            \draw (0.1, \pos) -- (-0.1, \pos) node[left] {\scriptsize $\label$};
        }
    \end{tikzpicture}\hfill
    % Subfigure 3: The Original Scaled Lattice (Scalar changed to 0.01)
    \begin{tikzpicture}[scale=0.55]
        % Clip box expanded slightly so outer labels don't get cut off
        \clip (-3.5, -3.5) rectangle (3.5, 3.5);
        
        % Scale factor = 0.01
        \foreach \x in {-5,...,5} {
            \foreach \y in {-5,...,5} {
                \fill[black] ({\x*0.005*245.647615159 - \y*0.005*139.08331632}, {\x*0.005*152.352384841 + \y*0.005*269.08331632}) circle (0.05);
            }
        }
        
        % Draw axes
        \draw[->, thick] (-3.5, 0) -- (3.5, 0);
        \draw[->, thick] (0, -3.5) -- (0, 3.5);

        \draw[red, thick, dashed] (0,0) -- ({-0.005*139.08331632},{0.005*269.08331632}) -- ({0.005*(-139.08331632+245.647615159)},{0.005*(269.08331632+152.352384841)}) -- ({0.005*245.647615159},{0.005*152.352384841}) -- cycle;
        
        % Unlabeled tick marks (100 and 200 mapped via scale factor 0.01 to 1 and 2)
        \foreach \pos in {-2, -1, 1, 2} {
            \draw (\pos, 0.1) -- (\pos, -0.1);
            \draw (0.1, \pos) -- (-0.1, \pos);
        }
        
        % Labeled tick marks at 300, -300 (300 * 0.01 = 3)
        \foreach \pos/\label in {-3/-600, 3/600} {
            \draw (\pos, 0.1) -- (\pos, -0.1) node[below] {\scriptsize $\label$};
            \draw (0.1, \pos) -- (-0.1, \pos) node[left] {\scriptsize $\label$};
        }
    \end{tikzpicture}
    
    \caption{$M\bZ^2$ (left) and two ``approximations'' $\Sigma(\beta\fa)$ using $r=20$ and $200$.}\label{fig:1}
\end{figure}
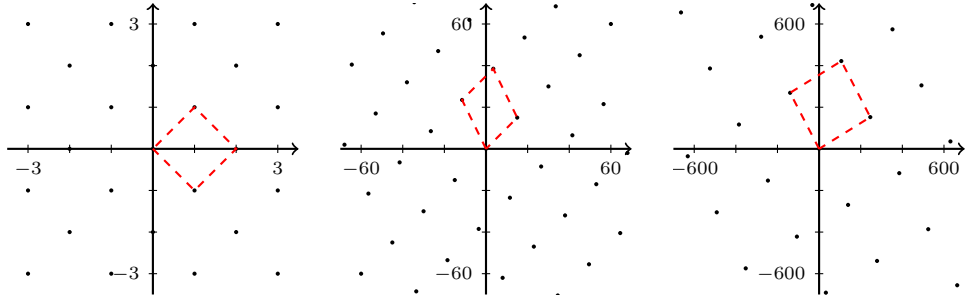

Figure~\ref{fig:1} illustrates the improvement from $r=20$ to $r=200$. The input lattice is the left-side image. In the middle image, $20\bc\approx(21.1,1.1)^T$ is approximated by $\Sigma(11+2\sqrt{34})\approx (22.7,-0.7)^T$, and the resulting transformation can be seen by how it affects the outlined fundamental region. The transformation is much closer to orthogonal in the right-side image, where $200\bc\approx(210.5,10.9)^T$ is approximated by $\Sigma(111+17\sqrt{34})\approx (210.1,11.9)^T$.

The reduction is complete: solving a lattice problem in the $\ell_2$ norm in $M\bZ^n$ is essentially equivalent to solving it in the canonical embedding of $\beta\fa$, an ideal in the monogenic order $\bZ[\lambda]$. To see why, observe that applying the orthogonal transformation $(EC)^{-1}$ and scaling by $r$ does not affect relative vector lengths, so $M\bZ^n$ may be replaced by $r(EC)^{-1}M\bZ^n=r(EC)^{-1}A\bZ^n$. But \begin{align}\label{eq:main}\nonumber&& r(EC)^{-1}A&=r(EC)^{-1}(AE)E^{-1} && \\\nonumber&& &= r(EC)^{-1}(ED)E^{-1}&& \text{by definition of }E\text{ and }D \\\nonumber && &=rC^{-1}DE^{-1} && \\\nonumber && &=rDC^{-1}E^{-1} && \text{since }D\text{ and }C^{-1}\text{ are diagonal}\\\nonumber && & = rD(EC)^T && \text{since }EC\text{ is orthogonal}\\ && &=rCDE^T && \text{since }C\text{ is diagonal.}\end{align} Now, $E^T\bZ^n$ is $\Sigma(\fa)$ by definition of $\fa$, and the diagonal of $rCD$ is approximately $\Sigma(\beta)$. Thus $rCDE^T\bZ^n$ approximates $\Sigma(\beta\fa)$, as desired.

In the example, $\bZ[\lambda]$ turned out to be the ring of integers in a degree-$n$ number field, namely the quadratic field $\bQ(\sqrt{34})$. In particular, $\fa$ was invertible and coprime to the trivial conductor. This is not uncommon experimentally (which distinguishes our approach from the reduction mentioned at the end of Section~\ref{ss:numbertheory}), but it is also not guaranteed. The primary remaining obstacle is to ensure an invertible ideal in a degree-$n$ number field. We offer two methods, one probabilistic and one deterministic, in Sections \ref{sec:4} and \ref{sec:5}.

\subsection{The algorithm} Forcing $[\bQ(\lambda):\bQ]=n$ and $\fa$ invertible is the black box in line~3. Since lines~7 and~10 below only make sense when $\bQ(\lambda)$ has $n$ real embeddings, the reader is asked to take this for granted throughout this section. Proofs can be found under Theorems~\ref{thm:sub2} and~\ref{thm:sub3}.

\begin{algorithm}
    \caption{Approximate an input lattice up to orthogonal transformation with some lattice $\Sigma(\fa)$ of an invertible ideal $\fa\subseteq\bZ[\lambda]$.}\label{alg:1}
    \DontPrintSemicolon
    \KwIn{$M\in\Mat_n(\bZ)$ (lattice basis) and $\kappa\in(1,\infty)$ (condition number bound)}
    \KwOut{a $\bZ$-basis for an ideal lattice}
    $S,U,V\gets$ Smith Normal Form matrices\commentR{$\triangleright\;MV=US$ with $S$ diagonal}
    $A\gets USU^T$\;
    $A,r\gets$ output of Subroutine~\ref{sub:2} or \ref{sub:3}\commentR{$\triangleright\;$input $A,\kappa$; finds a better $A$}
    $f(x)\gets \det(xI_n-A)$\;
    $\lambda\gets$ largest eigenvalue of $A$\;
    $(\alpha_1,\dots,\alpha_n)^T\gets$ last column of $\text{adj}(\lambda I_n-A)$\;
    $c_i\gets$ $i^\text{th}$ embedding length of $(\lambda\alpha_1,\dots,\lambda\alpha_n)$\commentR{$\triangleright\;$in $\ell_2$ norm; see Remark~\ref{rem:float}}
    $r\gets\smash{rn\lambda^{n-1}\!\max_i|c_i|(\sqrt{\kappa}+1)/(\sqrt{\kappa}-1)}$\;
    $V\gets$ matrix with $i^{\text{th}}$ column $\Sigma(\lambda^{i-1})$\commentR{$\triangleright\;$Vandermonde matrix}
    $\mathbf{c}\gets(c_1^{-1},\dots,c_n^{-1})^T$\;
    $\mathbf{b}\gets$ point in $\bZ^n$ with $|rV^{-1}\mathbf{c}-\mathbf{b}|_{\infty}<1$\commentR{$\triangleright\;$see Remark~\ref{rem:float}}
    $\beta\gets(1,\lambda,\dots,\lambda^{n-1})\mathbf{b}$\;
    \While(\commentF{$\triangleright\;$skip if using Subroutine~\ref{sub:2} and}){$1\not\in(\beta,f'(\lambda))$}{
        $\beta\gets\beta+1$\commentR{perhaps even \ref{sub:3}; see Remark~\ref{rem:coprime}}
    }
    \Return{$\Sigma(\alpha_1\beta),\dots,\Sigma(\alpha_n\beta)$}
\end{algorithm}

The input $\kappa$ controls quality of approximation between the initial generic lattice and the output ideal lattice. The manner in which $\kappa$ relates solutions to $\SVP$, $\SIVP$, $\uSVP$, and $\CVP$ in these lattices is the subject of Section~\ref{ss:condition}.

\begin{remark}\label{rem:float}Computing $\bb$ in line 11 requires floating point arithmetic. Entries of $\bb$ need only be within 1 of the corresponding entries of $rV^{-1}\bc$, so if $\mu$ is the maximum magnitude among entries of $\bc$ and $V^{-1}$, it suffices to compute $\bc$ and $V^{-1}$ to within $(2nr\mu)^{-1}$ of their true value. The total bit length of such a rational approximation is $O(\log (nr\mu))$, roughly the same as its integer truncation.\end{remark}

\begin{remark}\label{rem:coprime}Lines 13--14 make $\beta$ coprime to the conductor of $\bZ[\lambda]$, which holds automatically when Subroutine~\ref{sub:2} is used. Even if Subroutine~\ref{sub:3} is used, this \textbf{while} loop can be skipped if we only care that the output ideal is invertible (as a principal ideal is invertible whether of not it is coprime to $\ff(\bZ[\lambda])$). Also, note that $1\in (\beta,f'(\lambda))$ is stronger than $1\in (\beta,\ff(\bZ[\lambda]))$, but computing the conductor is hard.\end{remark}

Proofs of correctness for Algorithm~\ref{alg:1} must be postponed until after Subroutines~\ref{sub:2} and~\ref{sub:3}. For now, note that all steps outside of perhaps lines 3 and 13--14 can be performed in a polynomial (in the input bit length) number of operations. (Storjohann's deterministic algorithm computes $S$, $U$, and $V$ in line 1 in $O(n^4\log^2(n\mu))$ bit operations, where $\mu$ is the maximum entry magnitude of $M$. The resulting matrices $U$ and $V$ have entry magnitudes $(n\mu)^{O(n)}$ \cite[Section 8.1]{storjohann}.)

%\begin{proposition}\label{prop:A_bound}In $O(n^4\log^2(n\|M\|))$ bit operations, Storjohann's algorithm computes $S$, $U$, and $V$ in line 1 such that $A$ in line 2 satisfies $\log \|A\|=O(n\log(n\|M\|))$.\end{proposition}

%\begin{proof}The complexity $O(n^4\log^2(n\|M\|))$ is a simplified overestimate for what Storjohann proves in \cite{storjohann}. He also proves that his Smith multiplier $V$ satisfies $\|V\|=(n\|M\|)^{O(n)}$. Thus $\|U\|\leq\|US\|=\|MV\|\leq n\|M\|\|V\|=(n\|M\|)^{O(n)}$, which bounds the entries of $A$ in line 2: \[\|A\|=\|MVU^T\|\leq n\|MV\|\|U\|=(n\|M\|)^{O(n)}.\] Taking the logarithm completes the proof.\end{proof}

%There are two lattice approximations that occur in Algorithm~\ref{alg:1}. One of them is hidden in Subroutine~\ref{sub:2} or \ref{sub:3}. The other is when $r(a_1^{-1},\dots,a_n^{-1})^T$ is replaced by $\Sigma(\beta)$ in lines 10 and 11. Let us check that the condition number associated to the latter approximation is at most $\sqrt{\kappa}$. This leaves an additional $\sqrt{\kappa}$ of wiggle room for the subroutine's associated condition number.

%\begin{proposition}Given input $M$ and $\kappa$, let $A$ denote its value in line 5 and let $B$ be the matrix with columns output by Algorithm~\ref{alg:1}. If , then $\kappa_2(AB^{-1})<\sqrt{\kappa}$.\end{proposition}

\subsection{Approximation quality}\label{ss:condition} The role of $\kappa$ is to bound the condition number of a matrix that transforms the input lattice into the output lattice. Throughout this section, $p$ defines the norm.

\begin{definition}\label{def:condition}The \emph{operator norm} of $T\in\text{GL}_n(\bR)$ in the $\ell_p$ norm is \[\|T\|_p=\max_{|\bu|_p=1}|T\bu|_p,\] and the \emph{condition number} is \[\kappa_p(T)=\|T\|_p\|T^{-1}\|_p=\max_{\bu,\bv\neq 0}\frac{|T\bu|_p/|\bu|_p}{|T\bv|_p/|\bv|_p}.\]\end{definition}

The condition number controls the change in what ``short'' means in the input versus output lattice. We state the propositions below only for the $\ell_2$ norm, although they hold more generally mutatis mutandis. Both are certainly well-known. 

\begin{proposition}\label{prop:condition}If $A\bZ^n$ and $B\bZ^n$ are full-rank lattices, then multiplying by $AB^{-1}$ transforms a solution in $B\bZ^n$ to $\SVP_p$, $\SIVP_p$, $\uSVP_p$, or $\CVP_p$ for the parameter $\gamma$ into a solution in $A\bZ^n$ for $\kappa\gamma$, where $\kappa=\kappa_p(AB^{-1})$.\end{proposition}

\begin{proof}Let $T=AB^{-1}$ and $\kappa=\kappa_p(T)$. Consider $\gamma$-$\SVP_p$. Let $T\bv$ be a shortest vector in $A\bZ^n$ and suppose $\bu$ solves $\gamma$-$\SVP_p$ in $B\bZ^n$, implying $|\bu|_p\leq \gamma|\bv|_p$. Thus \begin{equation}\label{eq:condition}\frac{|T\bu|_p}{|T\bv|_p} \leq\frac{\gamma|T\bu|_p/|\bu|_p}{|T\bv|_p/|\bv|_p}\leq \kappa\gamma,\end{equation} meaning $T\bu$ solves $\kappa\gamma$-$\SVP_p$ in $A\bZ^n$. 

The exact same argument applies to the other lattice problems.\end{proof}

Since $\kappa=1$ cannot be achieved regardless of how large $r$ is in Algorithm~\ref{alg:1}, Proposition~\ref{prop:condition} suggests we cannot perfectly preserve the error tolerance $\gamma$ in lattice problems. That is not necessarily the case.

\begin{proposition}\label{prop:exact}Suppose $p\in\bZ\cup\{\infty\}$, $\Lambda$ is a rank $n$ integer lattice, and $\gamma^q=\frac{s}{t}$ for positive integers $q,s,t$ with $p\,|\,q$ if $p<\infty$. In the case of $\CVP_p$, suppose further that the target $\bx$ satisfies $d\bx\in\bZ$ with $0\neq d\in\bZ$. If \[\kappa^q<\begin{cases}1+(s\lambda_1(\Lambda)^q)^{-1} & \SVP_p,\uSVP_p\\1+(s\lambda_n(\Lambda)^q)^{-1} & \SIVP_p\\ 1+(s(dn\lambda_n(\Lambda))^q)^{-1} & \CVP_p,\end{cases}\] then a lattice problem solution in $\Lambda$ for the parameter $\kappa\gamma$ is also a solution for $\gamma$.\end{proposition}

\begin{proof}Let $\lambda_i=\lambda_i(\Lambda)$. Suppose $\kappa^q<1+(sn\lambda_1^q)^{-1}$ and that $\bv\in\Lambda$ solves $\kappa\gamma$-$\SVP_p$. Then \[|\bv|_p^q\leq (\kappa\gamma\lambda_1)^q<\left(1+\frac{1}{s\lambda_1^q}\right)\frac{s\lambda_1^q}{t}\leq \frac{s\lambda_1^q+1}{t}.\] Since $\lambda_1^q$ and $|\bv|_p^q$ are integers (recall $\Lambda\subseteq\bZ^n$) and there are no integers strictly between $\frac{1}{t}(s\lambda_1^q+1)$ and $\frac{1}{t}(s\lambda_1^q)=(\gamma\lambda_1)^q$, we conclude that $\bv$ solves $\gamma$-$\SVP_p$ in $\Lambda$.

The exact same argument applies to $\uSVP_p$ and $\SIVP_p$.

For $\CVP_p$ with target $\bx$ satisfying $d\bx\in\bZ^n$, if $\bv$ solves $\kappa\gamma$-$\CVP_p$, the argument above applies beginning with the integer $|d(\bv-\bx)|_p^q$. Note that $\Lambda$ has covering radius at most $n\lambda_n$.\end{proof}

\begin{lemma}\label{lem:condition}For any $A,B\in\textup{GL}_n(\bR)$, if \[\|A-B\|_p\leq\frac{\kappa-1}{\|A^{-1}\|_p(\kappa+1)}\] for some $\kappa\in\bR$, then $\kappa_p(AB^{-1})\leq\kappa$.\end{lemma}

\begin{proof}Let $T=BA^{-1}$. Observe that \[|T\bv-\bv|_p\leq \|T-I_n\|_p|\bv|_p\leq \|A^{-1}\|_p\|A-B\|_p|\bv|_p\leq \frac{\kappa-1}{\kappa+1}|\bv|_p,\] where the final inequality is the lemma's hypothesis. We may thus apply the triangle inequality to bound $|T\bv|_2$ from above and below: \[\left(1-\frac{\kappa-1}{\kappa+1}\right)\!|\bv|_p\leq|\bv|_p-|T\bv-\bv|_p\leq|T\bv|_p\leq |\bv|_p+|T\bv-\bv|_p\leq\left(1+\frac{\kappa-1}{\kappa+1}\right)\!|\bv|_p\] The ratio of the right-side factor to the left-side factor is exactly $\kappa$. Inverting $T$ does not change its condition number, so $\kappa_p(AB^{-1})\leq\kappa$.\end{proof}

\section{Randomized subroutine to find a ring of integers}\label{sec:4}

If $A$ is perturbed randomly, we expect its characteristic polynomial, call it $f(x)$, to be irreducible with probability approaching 1 as the dimension $n$ grows. This was confirmed by Eberhard for generic matrices \cite{eberhard} and Ferber, Jain, Sah, and Sawhney for symmetric matrices \cite{ferber}, both assuming the extended Riemann hypothesis. These results assert that characteristic polynomials of random matrices behave, at least in terms of irreducibility, like random monic, integer polynomials, which have been long known to be irreducible with high probability \cite{vander}. 

We also care about whether $\bZ[\lambda]\simeq\bZ[x]/(f(x))$ is the full ring of integers in $\bQ(\lambda)$. In this regard, there are no results on characteristic polynomials of randomly selected integer matrices. There is only a conjecture of Wang and Yu asserting that symmetric integer matrices should make $\text{disc}f(x)$ squarefree (which implies $\bZ[\lambda]$ is a ring of integers) a positive proportion of the time \cite[Remark 1]{wang}. There is however, an analogous result for polynomials. Bhargava, Shankar, and Wang prove that in any fixed degree, the proportion of monic, integer polynomials of bounded coefficient height with squarefree discriminant approaches \[\frac{1}{2}\prod_{p\geq 3}\left(1-\frac{3p-1}{p^2(p+1)}\right)\approx 0.307\] as the height bound grows \cite{bhargava}. They also prove that the proportion of polynomials that define a ring of integers approaches $\frac{6}{\pi^2}\approx 0.608$.

Like Wang and Yu, we conjecture that these positive-proportion results extend to the matrix setting. For $A\in\Mat_n(\bZ)$ and $r>0$, we use $B(A,r)$ denote the $\ell_{\infty}$ open ball of radius $r$ in $\bR^{n^2}$ around $A$.

\begin{conjecture}\label{conj}For $i=1,2$, there exists $p_i(n)>0$ such that the proportion $P_i(A,r)$ of symmetric matrices in $B(A,r)\cap\Mat_n(\bZ)$ with characteristic polynomial of squarefree discriminant ($P_1$) or trivial conductor ($P_2$) satisfies \[\left|P_i(A,r)-p_i(n)\right|<2^{n^c}r^{-c'}\] for any $A$ and $r$ and some absolute constants $c,c'>0$. Both sequences $\{p_i(n)\}_n$ are bounded below by positive constants.\end{conjecture}

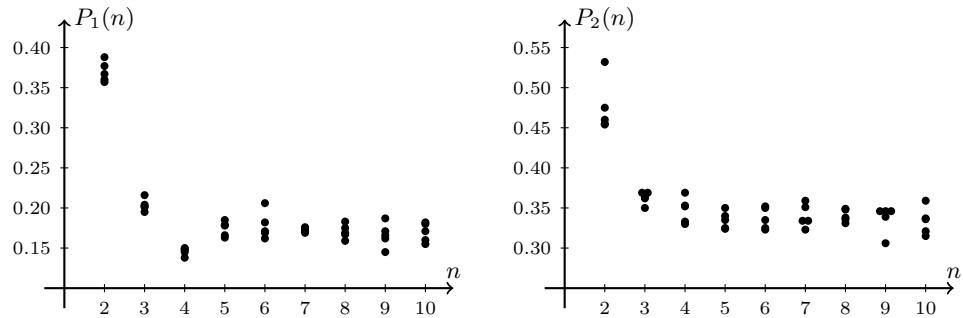
\begin{figure}[htbp]
    \centering
    % Subfigure 1: First Dot Plot (Zoomed in to 0.10 - 0.40)
    \begin{tikzpicture}[scale=0.53]
        % Clip box adjusted for the new axis position
        \clip (-0.4, -0.7) rectangle (11, 7.5);
        
        % Draw axes (y-axis placed at x=1)
        \draw[->, thick] (0.5, 0) -- (10.7, 0) node[above] {\small $n$};
        \draw[->, thick] (1, -0.5) -- (1, 6.7) node[right] {\small $P_1(n)$};
        
        % X-axis labeled tick marks
        \foreach \x in {2,3,4,5,6,7,8,9,10} {
            \draw (\x, 0.1) -- (\x, -0.1) node[below] {\scriptsize $\x$};
        }
        
        % Y-axis labeled tick marks (0.10 to 0.40, scaled dynamically)
        \foreach \y/\label in {1/0.15, 2/0.20, 3/0.25, 4/0.30, 5/0.35, 6/0.40} {
            \draw (1.1, \y) -- (0.9, \y) node[left] {\scriptsize $\label$};
        }

        % Plotting Data Points for Graph 1
        % Plotted y = (data_value - 0.10) * 20
        % n=2
        \fill[black] (2, 5.14) circle (0.1);
        \fill[black] (2, 5.20) circle (0.1);
        \fill[black] (2, 5.54) circle (0.1);
        \fill[black] (2, 5.76) circle (0.1);
        \fill[black] (2, 5.34) circle (0.1);
        % n=3
        \fill[black] (3, 2.32) circle (0.1);
        \fill[black] (3, 2.04) circle (0.1);
        \fill[black] (3, 1.90) circle (0.1);
        \fill[black] (3, 2.08) circle (0.1);
        \fill[black] (3, 2.02) circle (0.1);
        % n=4
        \fill[black] (4, 0.76) circle (0.1);
        \fill[black] (4, 0.94) circle (0.1);
        \fill[black] (4, 0.90) circle (0.1);
        \fill[black] (4, 0.98) circle (0.1);
        \fill[black] (4, 1.00) circle (0.1);
        % n=5
        \fill[black] (5, 1.70) circle (0.1);
        \fill[black] (5, 1.58) circle (0.1);
        \fill[black] (5, 1.56) circle (0.1);
        \fill[black] (5, 1.26) circle (0.1);
        \fill[black] (5, 1.32) circle (0.1);
        % n=6
        \fill[black] (6, 1.38) circle (0.1);
        \fill[black] (6, 1.24) circle (0.1);
        \fill[black] (6, 2.12) circle (0.1);
        \fill[black] (6, 1.64) circle (0.1);
        \fill[black] (6, 1.42) circle (0.1);
        % n=7
        \fill[black] (7, 1.42) circle (0.1);
        \fill[black] (7, 1.38) circle (0.1);
        \fill[black] (7, 1.52) circle (0.1);
        \fill[black] (7, 1.50) circle (0.1);
        \fill[black] (7, 1.46) circle (0.1);
        % n=8
        \fill[black] (8, 1.34) circle (0.1);
        \fill[black] (8, 1.38) circle (0.1);
        \fill[black] (8, 1.18) circle (0.1);
        \fill[black] (8, 1.66) circle (0.1);
        \fill[black] (8, 1.50) circle (0.1);
        % n=9
        \fill[black] (9, 0.90) circle (0.1);
        \fill[black] (9, 1.24) circle (0.1);
        \fill[black] (9, 1.74) circle (0.1);
        \fill[black] (9, 1.32) circle (0.1);
        \fill[black] (9, 1.42) circle (0.1);
        % n=10
        \fill[black] (10, 1.60) circle (0.1);
        \fill[black] (10, 1.20) circle (0.1);
        \fill[black] (10, 1.64) circle (0.1);
        \fill[black] (10, 1.42) circle (0.1);
        \fill[black] (10, 1.10) circle (0.1);
    \end{tikzpicture}\hfill
    % Subfigure 2: Second Dot Plot (Zoomed in to 0.30 - 0.55)
    \begin{tikzpicture}[scale=0.53]
        % Clip box adjusted for the new axis position
        \clip (-0.4, -1.7) rectangle (11, 6.5);
        
        % Draw axes (y-axis placed at x=1)
        \draw[->, thick] (0.5, -1) -- (10.7, -1) node[above] {\small $n$};
        \draw[->, thick] (1, -1.5) -- (1, 5.7) node[right] {\small $P_2(n)$};
        
        % X-axis labeled tick marks
        \foreach \x in {2,3,4,5,6,7,8,9,10} {
            \draw (\x, -0.9) -- (\x, -1.1) node[below] {\scriptsize $\x$};
        }
        
        % Y-axis labeled tick marks (0.30 to 0.55, scaled dynamically)
        \foreach \y/\label in {0/0.30, 1/0.35, 2/0.40, 3/0.45, 4/0.50, 5/0.55} {
            \draw (1.1, \y) -- (0.9, \y) node[left] {\scriptsize $\label$};
        }

        % Plotting Data Points for Graph 2
        % Plotted y = (data_value - 0.30) * 20
        % Overlaps are highlighted with larger circle radii
        
        % n=2
        \fill[black] (2, 3.20) circle (0.1);
        \fill[black] (2, 3.08) circle (0.1);
        \fill[black] (2, 4.64) circle (0.1);
        \fill[black] (2, 3.50) circle (0.1);
        \fill[black] (2, 3.10) circle (0.1);
        % n=3
        \fill[black] (2.93, 1.38) circle (0.1);
        \fill[black] (3.07, 1.38) circle (0.1); % 2 identical points (0.369)
        \fill[black] (3, 1.28) circle (0.1);
        \fill[black] (3, 1.24) circle (0.1);
        \fill[black] (3, 1.00) circle (0.1);
        % n=4
        \fill[black] (4, 0.66) circle (0.1);
        \fill[black] (4, 1.04) circle (0.1);
        \fill[black] (4, 0.60) circle (0.1);
        \fill[black] (4, 1.38) circle (0.1);
        \fill[black] (4, 1.06) circle (0.1);
        % n=5
        \fill[black] (5, 0.80) circle (0.1);
        \fill[black] (5, 0.70) circle (0.1);
        \fill[black] (5, 1.00) circle (0.1);
        \fill[black] (5, 0.48) circle (0.1);
        \fill[black] (5, 0.50) circle (0.1);
        % n=6
        \fill[black] (6, 0.70) circle (0.1);
        \fill[black] (6, 0.50) circle (0.1);
        \fill[black] (6, 1.00) circle (0.1);
        \fill[black] (6, 0.46) circle (0.1);
        \fill[black] (6, 1.04) circle (0.1);
        % n=7
        \fill[black] (6.93, 0.68) circle (0.1);
        \fill[black] (7.07, 0.68) circle (0.1);% 2 identical points (0.334)
        \fill[black] (7, 0.46) circle (0.1);
        \fill[black] (7, 1.02) circle (0.1);
        \fill[black] (7, 1.18) circle (0.1);
        % n=8
        \fill[black] (8, 0.76) circle (0.1);
        \fill[black] (8, 0.62) circle (0.1);
        \fill[black] (8, 0.98) circle (0.1);
        \fill[black] (8, 0.96) circle (0.1);
        \fill[black] (8, 0.72) circle (0.1);
        % n=9
        \fill[black] (9, 0.12) circle (0.1);
        \fill[black] (9, 0.78) circle (0.1);
        \fill[black] (8.86, 0.92) circle (0.1);
        \fill[black] (9, 0.92) circle (0.1);
        \fill[black] (9.14, 0.92) circle (0.1);% 3 identical points (0.346)
        % n=10
        \fill[black] (10, 1.18) circle (0.1);
        \fill[black] (10, 0.30) circle (0.1);
        \fill[black] (10, 0.74) circle (0.1);
        \fill[black] (10, 0.72) circle (0.1);
        \fill[black] (10, 0.42) circle (0.1);
    \end{tikzpicture}
    
    \caption{Probability of squarefree discriminant (left) and trivial conductor.}\label{fig:2}
\end{figure}

Factoring is the natural bottleneck for any test of Conjecture~\ref{conj}. Even if we considered matrices from an $\ell_{\infty}$ unit ball around the origin in $\bR^{n^2}$, discriminants approach 900 bits when $n=15$. For this reason, we present data only up to dimension $10$ in Figure~\ref{fig:2}. In each dimension $n$, we test 5 centers $A$ for $B(A,r)$ by choosing matrix entries uniformly from $(-2^{100/n^2},2^{100/n^2})\cap \bZ$ (insisting on matrix symmetry). The radius $r=2^{50/n^2}$ defines a ball around each center from which we choose a random sample of $1000$ symmetric integer matrices. The proportion of matrices whose characteristic polynomial has squarefree discriminant is on the left (one dot per sampled ball), and the proportion of those that define a full ring of integers is on the right. 

Remark that when $n=2$, it is straightforward to prove the exact limits of $P_1(A,r)$ and $P_2(A,r)$ as $r$ grows. Both the discriminant formula of a symmetric matrix and the criterion for orders being maximal are far simpler when $n=2$ compared to $n\geq 3$. Multiplying local probabilities gives \[\lim_{r\to\infty}P_1(A,r)=\frac{1}{2}\! \prod_{p\equiv 3\,\text{mod}\,4} \!\!\left(1-\frac{1}{p^2}\right)\!\! \prod_{p\equiv 1\,\text{mod}\,4} \!\!\left(1-\frac{3p-2}{p^3}\right) \approx 0.366\] and \[\lim_{r\to\infty}P_2(A,r)=\frac{5}{4}\lim_{r\to\infty}P_1(A,r)\approx 0.458,\] both uniformly in $A$.

For any $n$, the local factors appearing in the product should be bounded below (and above) by $1-O(\frac{1}{p^2})$. Discriminant divisibility by $p$ defines a codimension-1 subvariety in the space of symmetric matrices over $\bF_p$. That is, roughly $\frac{1}{p}$ points in the subvariety meet this condition. For nonsingular points, lifting from mod$\,p$ to mod\,$p^2$ imposes another $1/p$ density decrease, while the singular locus already has codimension 2 over $\mathbb{F}_p$. Hence the failure rate is $O(1/p^2)$. The author has not verified formally what local factors equal, but in the Euler product for $P_2(A,r)$ they appear to approach $1-\frac{2}{p^2}$. Hence the predicted probability of \[\prod_{p\text{ prime}}\!\!\left(1-\frac{2}{p^2}\right)\approx 0.323\] in Remark \ref{rem:random}.

Finally, a note on the bound $2^{n^{c}}r^{-c'}$ in Conjecture~\ref{conj}: The denominator $r^{c'}$ matches the shape of Bhargava, Shankar, and Wang's error bound, in which $c'=\frac{1}{5}-\varepsilon$ \cite[Page 2 and Theorem 4.4]{bhargava}. (The value of $c'$ balances errors from the sieve's tail and main body.) Regarding the numerator $2^{n^c}$, such expressions hide behind big-$O$ notation in \cite{bhargava}. We care to be explicit, however, because convergence rate impacts algorithm running time. The shape of our bound is justified by identifying the fastest growing hidden constant in \cite{bhargava}, which appears to come from Bhargava's foundational work \cite{bhargava2}. In his proof of Theorem~3.3 (the quantitative Ekedahl sieve \cite{ekedahl}), the largest contributor to the final constant is a product of single-variable degrees of resultants built from an initial set of multivariate polynomials (see the third paragraph and the end of the proof of Lemma~3.1 in \cite{bhargava2}). In our setup (where Bhargava's ``$k$'' is 2 and the dimension is $\frac{1}{2}n(n+1)$), this product turns out to be $\smash{2^{O(n^3)}}$. 

\begin{algorithm}
\SetAlgorithmName{Subroutine}{Subroutine}{}
    \caption{Probabilistic method to adjust $A$ so that $\det(xI_n-A)$ is irreducible and $\bZ[\lambda]$ is the ring of integers in $\bQ(\lambda)$.}\label{sub:2}
    \DontPrintSemicolon
    \KwIn{A symmetric $A\in\Mat_n(\bZ)$ and a real number $\kappa>1$}
    \KwOut{A scaled, perturbed copy of $A$}
    $B\gets A$\commentR{$\triangleright\;$second copy of $A$}
    $r\gets (\sqrt{\kappa}-1)/(n\|A^{-1}\|_2(\sqrt{\kappa}+1))$\;
    $\lambda\gets$ any eigenvalue of $A$\;
    \While(\commentF{$\triangleright\;$see Remark~\ref{rem:conductor}}){$[\bQ(\lambda):\bQ]\neq n$\textup{ or }$1\not\in\ff(\bZ[\lambda])$}{
        $r\gets 2r$\;
        $A\gets 2A$\;
        $B\gets$ symmetric with $\max |(B-A)_{i,j}|< r$\commentR{$\triangleright\;$chosen uniformly at random}
        $\lambda\gets$ any eigenvalue of $B$\;
    }
    \Return{$B,1$}
\end{algorithm}

\begin{remark}\label{rem:conductor}There are no known (non-quantum) polynomial time algorithms for testing whether $\bZ[\lambda]$ is a ring of integers or has squarefree discriminant. By a result of Lenstra \cite[Theorem 4.4]{lenstra}, these problems are polynomial-time equivalent to finding the largest square factor of an integer.\end{remark}

\begin{theorem}If Conjecture~\ref{conj} is true, then Algorithm~\ref{alg:1} with Subroutine~\ref{sub:2} runs in bounded-error quantum polynomial time.\end{theorem}

\begin{proof}The probability that $B$ in line 7 satisfies the \textbf{while} loop condition is denoted $P_2(A,r)$ in Conjecture~\ref{conj}. The conjecture asserts that  implies that the probability of success exceeds $\frac{1}{2}p_2(n)$ after \textbf{while} loop$\bZ[\lambda]$ being the ring of integers in a degree-$n$ number field is bounded below by $\frac{1}{2}$ after a constant number of iterations of Subroutine~\ref{sub:2}'s \textbf{while} loop. This is exactly the criterion for terminating the \textbf{while} loop. To test whether $1\in\ff(\bZ[\lambda])$ in line 4, we use Lenstra's deterministic polynomial time (non-quantum) algorithm to reduce to the problem of computing the largest squarefree divisor of an integer \cite[Theorem 4.4]{lenstra}, which is solved in quantum polynomial time by Shor's factoring algorithm \cite{shor}.

All other steps in Algorithm~\ref{alg:1} and Subroutine~\ref{sub:2} evidently run in deterministic polynomial time in their input bit lengths, which are bounded by a polynomial in the bit length of $M$ and $\kappa$.\end{proof}

\begin{theorem}\label{thm:sub2}From Algorithm~\ref{alg:1} with Subroutine~\ref{sub:2}, $\alpha_1\beta,\dots,\alpha_n\beta$ is a $\bZ$-basis for an ideal $\fb\subseteq \bZ[\lambda]$, the ring of integers in the totally real field $\bQ(\lambda)$. Furthermore, if $B$ is the matrix with $i^\text{th}$ column $\Sigma(\alpha_i\beta)$ and $A$ is its original value in line 2, then $AB^{-1}\Sigma(\fb)=M\bZ^n$, and $\kappa_2(AB^{-1})<\kappa$.\end{theorem}

\begin{proof}There are two transformations in Algorithm~\ref{alg:1} that we handle separately. The first is from (the lattices generated by) $A_1\coloneqq A$ to $B_1$, where $B_1$ denotes the output of Subroutine~\ref{sub:2}. The second is from $A_2\coloneqq B_1$ to $B_2\coloneqq B$. We aim to bound each $\kappa_2(A_iB_i^{-1})$ by $\sqrt{\kappa}$, which will imply the claim: \[\kappa_2(AB^{-1})=\kappa_2((A_1B_1^{-1})(A_2B_2^{-1}))\leq \kappa_2(A_1B_1^{-1})\kappa_2(A_2B_2^{-1})<\kappa.\] 

To bound $\kappa_2(A_1B_1^{-1})$, let $j$ denote the number of iterations that were needed in the \textbf{while} loop in Subroutine~\ref{sub:2}. If $r$ is initial value in line 2, then \begin{equation}\label{eq:A1B1}\|A_1-2^{-j}B_1\|_2\leq n\!\max_{1\leq i,j\leq n}\!|(A_1-2^{-j}B_1)_{i,j}|<nr =\frac{\sqrt{\kappa}-1}{\|A_1^{-1}\|_2(\sqrt{\kappa}+1)},\end{equation} where the second inequality is line 7. By Lemma~\ref{lem:condition}, we have \[\sqrt{\kappa}>\kappa_2(A_1(2^{-j}B_1)^{-1})=\kappa_2(A_1B_1^{-1}).\]

Let $E$ be the matrix whose $i^\text{th}$ column is the $i^\text{th}$ embedding of $(\alpha_1,\dots,\alpha_n)$, which are eigenvectors of $A_2$, and let $D$ be the corresponding diagonal matrix of eigenvalues. By comparing the columns corresponding to $\lambda$ on either side of $A_2E=ED$, we see that the $\alpha_i$ do indeed form a $\bZ$-basis for an ideal in $\bZ[\lambda]$, which is the ring of integers in $\bQ(\lambda)$ by line 4 of Subroutine~\ref{sub:2}. 

Let $C$ be the diagonal matrix with $i^\text{th}$ diagonal entry $a_i^{-1}$. Since $A_2=B_1$ is symmetric by line 7 of Subroutine~\ref{sub:2}, $EC$ is an orthogonal matrix by the spectral theorem. Furthermore, $\bQ(\lambda)$ is totally real. 

Now let $C'$ denote the diagonal matrix with the embeddings of $\beta$ along its diagonal, and reassign $r$ to be its value after line 8 of Algorithm~\ref{alg:1}. Observe that the output matrix is $B_2=C'E^T$. We have \begin{align*}&& \kappa_2(A_2B_2^{-1})&=\kappa_2(r(EC)^{-1}A_2B_2^{-1}) && \text{(scaling/orthogonal transformation)}\\ && &=\kappa_2(rCDE^TB_2^{-1}) && \text{by (\ref{eq:main})}\\ && &=\kappa_2(C(r^{-1}C')^{-1}) && \text{by definition of }B_2.\end{align*} 

The distance between entries of $\bb$ and $rV^{-1}\ba$ is less than 1 by line 11. Hence the distance between entries of $\Sigma(\beta)=V\bb$, which forms the diagonal of $C'$, and $V(rV^{-1}\ba)$, which forms the diagonal of $rC$, is less than $n\|V\|=n\lambda^{n-1}$ (because $\lambda$ is the largest eigenvalue of $A_2$). Thus \begin{equation}\label{eq:CC'}\|C-r^{-1}C'\|_2<r^{-1}n\lambda^{n-1}=\frac{\sqrt{\kappa}-1}{\max\limits_{i\leq n}|a_i|(\sqrt{\kappa}+1)}=\frac{\sqrt{\kappa}-1}{\|C^{-1}\|_2(\sqrt{\kappa}+1)}.\end{equation} By Lemma~\ref{lem:condition} again, $\kappa_2(C(r^{-1}C')^{-1})<\sqrt{\kappa}$.\end{proof}

\section{Deterministic subroutine to find an invertible ideal}\label{sec:5}

In this section, we provide a polynomial time subroutine that guarantees an irreducible characteristic polynomial for $A$ from Algorithm~\ref{alg:1} as well as an invertible output ideal. There is no need to consider general $\ell_p$ norms here, so $p$ is often used to denote a prime number. 

\subsection{Polynomial irreducibility}\label{ss:irreducible} %The simplest way to obtain an irreducible characteristic polynomial is to make it irreducible modulo some fixed integer. An Eisenstein polynomial at some prime $p$ is the easiest target, but these have discriminant divisible by $p$. For sieve theoretic reasons explained in the next section, we wish to avoid small primes in our discriminant. So instead, we use Shoup's deterministic, polynomial time algorithm to find irreducible polynomials mod$\,p$ \cite{shoup}.

%\begin{definition}\label{def:Q}For primes $p<4n^2$, let $C_p\in\text{GL}_n(\bF_p)$ be the companion matrix the irreducible polynomial mod$\,p$ output by Shoup's algorithm. Let $q$ be the product of these primes, and let $C_q\in\text{Mat}_n(\bZ)$ be a lift of all $C_p$ that satisfies $\|C_q\|\leq\frac{q}{2}$.\end{definition} Note that $C_q$ takes the form \[\begin{bmatrix}& 1 & & \\ & & \ddots & \\ & & & 1 \\ * & * & \cdots & *\end{bmatrix}.\]

%\begin{lemma}If $B\in\textup{Mat}_n(\bZ)$ with $B\equiv C_q\,\text{mod}\,q$, then $\det(xI_n-B)$ is irreducible with discriminant coprime to $q$.\end{lemma}

%\begin{proof}The discriminant of $\det(xI_n-B)$ reduced $\text{mod}\,p$ equals that of $\det(xI_n-C_p)$. Finite fields are perfect, so irreducible implies nonzero discriminant.\end{proof}

The simplest way to obtain an irreducible characteristic polynomial is to fix a target irreducible polynomial $f(x)$ modulo some small prime $p$ and find a symmetric matrix $S_p$ with $\text{char}_{S_p}(x)\equiv f(x)\,\text{mod}\,p$. We can enforce that the final output matrix $A$ produced by the subroutine satisfies $A\equiv S_p\,\text{mod}\,p$. We briefly review the known efficient techniques to achieve this with $p=2$. 

Shoup's deterministic algorithm finds an irreducible $f(x)\in\bF_2[x]$ of degree $n$ \cite{shoup}. (The symmetry constraint makes it hard to use polynomials that are Eisenstein at $p$.) Let $C_f$ be its companion matrix satisfying $f(x)=x^n+\sum_i(C_f)_{i+1,n}x^i$.

A standard technique (see \cite{basu}, for example) to produce a symmetric (Hankel) matrix $H$ satisfying $HC_f=C_f^TH$ is letting $H_{i,j}=\text{Tr}(x^{i+j-2})$, the trace from $\bF_2[x]/(f(x))$ to $\bF_2$. In other words, $H$ represents the trace pairing with respect to the basis $1,x,\dots,x^{n-1}$. For example, when $f(x)=x^3+x+1$, the traces of $1$, $x$, $x^2$, $x^3$, and $x^4$ are 1, 0, 0, 1, and 0 (the antidiagonals of $H$ below). And indeed, \[HC_f=\begin{bmatrix}1 & 0 & 0 \\ 0 & 0 & 1 \\ 0 & 1 & 0\end{bmatrix}\begin{bmatrix}0 & 0 & 1 \\ 1 & 0 & 1 \\ 0 & 1 & 0\end{bmatrix}=\begin{bmatrix}0 & 1 & 0 \\ 0 & 0 & 1 \\ 1 & 1 & 0\end{bmatrix}\begin{bmatrix}1 & 0 & 0 \\ 0 & 0 & 1 \\ 0 & 1 & 0\end{bmatrix}=C_f^TH.\] This identity holds because $C_f$ represents multiplication by $x$ in $\bF_2[x]/(f(x))$, which is self-adjoint  with respect to the trace bilinear form. The equation above simply expresses that the trace is self-adjoint with respect to multiplication by $x$ (as detailed in the proof of Lemma~\ref{lem:S2}). 

The trace pairing is a non-alternating form. By Dickson's classification of symmetric bilinear forms (including over $\bF_2$) \cite{dickson}, this implies the trace pairing is diagonalizable, meaning $P^THP=D$ for some diagonal matrix $D$ and change of basis matrix $P$. Over $\bF_2$, $D=I_n$ is forced, so $H=(PP^T)^{-1}$. Thus $HC_f=C_f^TH$ reorganizes to $P^{-1}C_fP=(P^{-1}C_fP)^T$. We have found a symmetric matrix with characteristic polynomial $f(x)$. In the example, $1+x,\,1+x^2,\,1+x+x^2$ is an orthonormal basis with respect to the trace form, and indeed, \[P^{-1}C_f P=\begin{bmatrix}1 & 0 & 1 \\ 1 & 1 & 0 \\ 1 & 1 & 1\end{bmatrix}\begin{bmatrix}0 & 0 & 1 \\ 1 & 0 & 1 \\ 0 & 1 & 0\end{bmatrix}\begin{bmatrix}1 & 1 & 1 \\ 1 & 0 & 1 \\ 0 & 1 & 1\end{bmatrix}=\begin{bmatrix}1 & 1 & 0 \\ 1 & 1 & 1 \\ 0 & 1 & 0\end{bmatrix}\] is a asymmetric matrix with characteristic polynomial $x^3+x+1$.

Over fields of characteristic other than 2, the Gram--Schmidt algorithm diagonalizes $H$. Over fields of characteristic 2, non-degenerate does not imply non-alternating, so Gram--Schmidt alone can fail if it stumbles into a subspace in which the bilinear form is alternating. Still, there are simple algorithms to bypass this obstacle. See, for example, the StackExchange post by GreginGre \cite{stack}, which is summarized in the last paragraph of the proof below.

\begin{lemma}\label{lem:S2}A symmetric matrix in $S_2\in \textup{GL}_n(\bF_2)$ with irreducible characteristic polynomial can be found in deterministic polynomial time.\end{lemma}

\begin{proof}Shoup's algorithm finds an irreducible $f(x)\,\text{mod}\,2$ of degree $n$ in polynomial time. The matrix symmetric $H$ representing the trace pairing from $\bF_2[x]/(f(x))$ to $\bF_2$ with respect to the basis $1,x,\dots,x^{n-1}$ can be computed in $O(n^2)$ bit operations. Indeed, if $c_i$ is the coefficient of $x^{n-i}$ in $f(x)$, Newton's formulas assert \[\text{Tr}(x^k)=kc_k+\sum_{i=1}^{k-1}c_{k-i}\text{Tr}(x^i).\] So the $2n-1$ antidiagonals of $H$ can be computed recursively, each in $O(n)$ bit operations.

Observe that \begin{align*} && (HC_f)_{i,j}&=\sum_{k=1}^n\text{Tr}(x^{i+k-2})(C_f)_{k,j}&& \text{by definition of }H\\ && &=\text{Tr}\!\left(\!x^{i-1}\sum_{k=1}^n(C_f)_{k,j}x^{k-1}\!\right) && \text{by bilinearity of the trace} \\ && &=\text{Tr}(x^{i-1}\cdot (x\cdot x^{j-1}))&&\text{since }C_f\text{ represents multiplication by }x\\ && &=\text{Tr}(x^{j-1}\cdot (x\cdot x^{i-1})) && \\ && &=\text{Tr}\!\left(\!x^{j-1}\sum_{k=1}^n(C_f)_{k,i}x^{k-1}\!\right)&&\text{since }C_f\text{ represents multiplication by }x \\ && &=\sum_{k=1}^n(C_f^T)_{i,k}\text{Tr}(x^{j+k-2}) && \text{by bilinearity of the trace}\\ && &=(C_f^TH)_{i,j} && \text{by definition of }H.\end{align*} In particular, diagonalizing $H$ with some matrix $P$ completes the process as previously described---if $P^THP=I_n$, then $S_2\coloneqq P^{-1}C_fP$ is symmetric.

Since finite field extensions are separable, the trace pairing is a non-degenerate. Furthermore, the traces of $1\cdot 1,x\cdot x,\dots,x^{n-1}\cdot x^{n-1}$ cannot all vanish because these monomials are a basis for $\bF_2[x]/(f(x))$ (apply the Frobenius automorphism to $1,x,\dots,x^{n-1}$). In other words, the trace pairing is nonalternating. 

To find a diagonalizing basis $g_1(x),\dots,g_n(x)$, we follow \cite{stack}: Let $V_n=\bF_2[x]/(f(x))$, and pick any $g(x)\in V$ such that $\text{Tr}(g(x)^2)=1$. If the trace form remains nonalterating on $\langle g(x)\rangle^\perp$, let $g_1(x)\coloneqq g(x)$ and $V_{n-1}\coloneqq\langle g_1(x)\rangle^\perp$. Otherwise, pick any $h_1(x),h_2(x)\in \langle g(x)\rangle^\perp$ such that $\text{Tr}(h_1(x)h_2(x))=1$ (such polynomials exist since the trace is nondegenerate on $\langle g(x)\rangle^\perp$), and set $g_1(x)\coloneqq g(x)+h_1(x)+h_2(x)$ and $V_{n-1}\coloneqq\langle g_1(x)\rangle^\perp$. Note that $g(x)+h_1(x)\in \langle g_1(x)\rangle^\perp$ and $\text{Tr}((g(x)+h_1(x))^2)=1$, allowing the procedure to continue on $V_{n-1}$.\end{proof}

\subsection{Ideal invertibility}\label{ss:invertible}This is the primary challenge of the deterministic subroutine for Algorithm~\ref{alg:1}. The ring at hand is $\cO\coloneqq\bZ[x]/(f(x))$, where $f(x)=\det(xI_n-A)$, and the ideal that must be invertible is generated by the last column entries of $\text{adj}(xI_n-A)$. Our strategy is to make this ideal coprime to $f'(x)$, which is sufficient but not necessary for invertibility (see Theorem \ref{thm:sub3}). This turns out to be possible through a symmetric perturbation of $A$ that is small enough to negligibly affect the shape of the lattice. Specifically, our perturbation takes the form \begin{equation}\label{eq:ajbj}A+\begin{bmatrix} & a_2 & & \\ a_2 & b_2 & \ddots & \\ & \ddots & \ddots & a_n \\ & & a_n & b_n\end{bmatrix}.\end{equation} 

\begin{definition}\label{def:Aj}Let $A_j$ denote the top-left $j\times j$ minor of $A$, let $f_j(x)=\det (xI_j-A_j)$, let $\fa_j\subseteq\bZ[x]$ denote the ideal generated by the last column entries of $\text{adj}(xI_j-A_j)$, and call the top entry in this column $g_{j-1}(x)$.\end{definition} 

The search for $a_j$ and $b_j$ that make $(\fa_n,f'_n(x))=\bZ[x]$ is performed sequentially. Assume $a_2,b_2,\dots,a_{j-1},b_{j-1}$ have already been chosen so that $(\fa_{j-1},f'_{j-1}(x),p)=\bZ[x]$ for any prime $p$. We first choose $a_j$ to ensure that $(\fa_j,f'_{j-1}(x),p)=\bZ[x]$ for all $p$. To see why this is possible, observe that the top $j-1$ generators of $\fa_j$ can be computed as the product \begin{equation}\label{eq:cramer}\text{adj}(xI_{j-1}-A_{j-1})\begin{bmatrix}A_{1,j}\\\vdots\\A_{j-1,j}+a_j\end{bmatrix}.\end{equation} For any prime $\fp\,|\,p$ in $\bZ[x]$ that contains $f'_{j-1}(x)$, we have assumed that the final entry in at least one row of $\text{adj}(xI_{j-1}-A_{j-1})$ (these final entries are the coefficients of $a_j$ above) is not in $\fp$. Therefore, there is at most one bad congruence class $a_j\,\text{mod}\,p$ associated to $\fp$, where $p\in\bZ$ is the prime below $\fp$. We simply increment $a_j$ until it avoids every bad congruence class from the finitely many ramified primes with odd norm in $\bZ[x]/(f_{j-1}(x))$.

Once $a_j$ is chosen, the purpose of $b_j$ is to switch rings; to complete the induction, we need $\fa_j$ to avoid ramified primes with odd norm in $\bZ[x]/(f_j(x))$, not in $\bZ[x]/(f_{j-1}(x))$. By cofactor expansion along the bottom row of $xI_j-A_j$, \[f_j(x)=(x-A_{j,j}-b_j)f_{j-1}(x)+\sum_{i=1}^{j-1}A_{j,i}\text{adj}(xI_j-A_j)_{i,j}.\] This is a linear combination of the generators of $\fa_j$. Hence, if $\fp$ is a prime in $\bZ[x]$ that contains $\fa_j$, then $\fp$ automatically contains $f_j(x)$. The job of $b_j$ is to make sure $\fp$ does not also contain the derivative $f_j'(x)$, lest it be ramified in $\bZ[x]/(f_j(x))$. Since \[f_j'(x)=(x-A_{j,j}-b_j)f_{j-1}'(x)-(A_{j,j}+b_j)f_{j-1}(x)+\sum_{i=1}^{j-1}A_{j,i}\text{adj}(xI_j-A_j)'_{i,j},\] and since $a_j$ was specifically chosen so that $f_{j-1}'(x)$ is not in $\fp$, there is at most one bad congruence class $b_j\,\text{mod}\,p$ associated to $\fp$. As before, we increment $b_j$ until it avoids every bad congruence class from the finitely many odd-norm primes dividing $\fa_j$. This completes the inductive step.

There are two potential issues that can arise in the search for $a_j$ or $b_j$. First, since $f_{j-1}(x)$ can have up to $j-1$ roots mod$\,p$, if $p<j$, it is possible that there is no acceptable congruence class mod$\,p$ for $a_j$ or $b_j$; the search would never terminate. In particular, we must control $A$ modulo every prime less than $n$. In fact, we control $A$ at some larger primes as well.

\begin{definition}\label{def:Sq}For $p>2$, define $S_p\in\textup{Mat}_n(\bF_p)$ to be \[\begin{bmatrix}1 & 1 & & \\ 1 & & \ddots & \\ & \ddots & & 1 \\ & & 1 &\end{bmatrix}\hspace{0.5cm}\text{or}\hspace{0.5cm}\begin{bmatrix} & 1 & & \\ 1 & & \ddots & \\ & \ddots & & 1 \\ & & 1 &\end{bmatrix},\] the former if $p\,|\,n+1$, the latter if $p\nmid n+1$. Let $q$ be the product of all primes less than $4n^2$, and fix some $S_q\in\text{Mat}_n(\bZ)$ with maximum entry magnitude at most $\frac{q}{2}$ and $S_q\equiv S_p\,\text{mod}\,p$ for all $p\,|\,q$ (even $p=2$ using $S_2$ from Lemma~\ref{lem:S2}).\end{definition}

\begin{lemma}\label{lem:disc}The discriminant of characteristic polynomial of $S_q$ is coprime to $q$.\end{lemma}

\begin{proof}The discriminant of $\det(xI_n-S_q)$ reduced $\text{mod}\,p$ equals that of $\det(xI_n-S_p)$. When $p=2$, this polynomial is irreducible by Lemma~\ref{lem:S2}, and it therefore has a nonzero discriminant in $\bF_2$ (a perfect field).

For odd $p$ that do not divide $n+1$, $\det(xI_n-S_p)$ is the Chebyshev polynomial of the second kind, denoted $U_n(x/2)$, which are defined by the $U_0(x/2)=1$, $U_1(x/2)=x$, and $U_n(x/2)=xU_{n-1}(x/2)-U_{n-2}(x/2)$. The well-known substitution $x=y+y^{-1}$ makes \[U_n(x/2)=\frac{y^{n+1}-y^{-n-1}}{y-y^{-1}}=\frac{y^{-n}(y^{2n+2}-1)}{y^2-1}.\] In particular, the roots of $U_n(x/2)$ are $x=\zeta^i+\zeta^{-i}$ for $i=1,\dots,n$ and $\zeta$ a primitive $(2n+2)^\text{th}$ root of unity. These roots are distinct in $\overline{\bF}_p$ when $p\nmid n+1$.

For odd $p$ that divide $n+1$, $\det(xI_n-S_p)=U_n(x/2)-U_{n-1}(x/2)$. Using $x=y+y^{-1}$ again, this becomes \[\frac{(y^{n+1}-y^{-n-1})-(y^n-y^{-n})}{y-y^{-1}}=\frac{y^{-n}(y^{2n+1}+1)}{y+1}.\] Thus the roots are $x=\zeta^{2i-1}+\zeta^{-2i+1}$ for $i=1,\dots,n$ and $\zeta$ a primitive $(4n+2)^{\text{th}}$ root of unity. These roots are distinct in $\overline{\bF}_p$ when $p\nmid 2n+1$, which is true when $p\,|\,n+1$.\end{proof}

\begin{remark}It follows from this proof that odd primes $p<4n^2$ split completely in the ring $\bZ[\lambda]$ produced in Algorithm~\ref{alg:1}. Alternative choices for $S_p$ in Definition~\ref{def:Sq} can produce whatever splitting behavior we prefer. (Insisting that $S_p$ is symmetric is not strictly necessary. It just makes our analysis of Algorithm~\ref{alg:1} cleaner.)\end{remark}

The reason for the additional restriction at primes between $n$ and $4n^2$ stems from the second issue: the search interval required to find $a_j$ and $b_j$ could be too large. While the total number of congruence classes modulo primes that $a_j$ and $b_j$ must avoid is easily bounded by a polynomial in $n\log\|M\|$, sieve theory requires a stronger hypothesis. The sifting density, call it $\rho$ (unfortunately ``$\kappa$'' in the literature, which we reserve for the condition number), is an upper bound on the average number of congruence classes $w(p)$ to be avoided per prime $p$ in any interval: \[\sum_{y<p<z}\frac{w(p)}{p}=\rho\log\!\left(\frac{\log z}{\log y}\right)+O(1)\] for any $z > y\geq 2$. Enforcing $w(p)=0$ when $p<4n^2$ significantly reduces $\rho$ in our setup. Still, as soon as the interval bound $y$ reaches $4n^2$, $w(p)$ can spike up to $j$ if we are unlucky. Since sieves place $\rho$ in the exponent of our search intervals for $a_j$ and $b_j$, this is a big problem for the desired polynomial running time. We might overcome this exponent by distributing the huge search interval over all $j$ entries in the final column of $A_j$ rather than just the bottom two entries, but searching through a $j$-dimensional box is also inherently exponential in running time.

The solution is to impose an additional constraint on $a_j$, asking that it ``prepares'' a nice sieve problem for $b_j$. Similarly, $b_j$ must prepare the sieve for $a_{j+1}$. Recall that $b_j$ avoids congruence classes at primes dividing $[\bZ[x]:\fa_j]$. If $[\bZ[x]:\fa_j]$ were a randomly selected integer from some large interval, the expected sifting density would be $\kappa=0$, meaning \begin{equation}\label{eq:sum_norm}\sum_{p\text{ prime}}\frac{1}{p}v_p([Z[x]:\fa])=O(1),\end{equation} where $v_p$ is the valuation at $p$. This is precisely the additional constraint on $a_j$. Similarly, since $a_{j+1}$ must avoid congruence classes at primes dividing $\text{disc}f(x)$, we ask $b_j$ to enforce \begin{equation}\label{eq:sum_disc}\sum_{p\text{ prime}}\frac{1}{p}v_p(\text{disc}f(x))=O(1).\end{equation} The cost is a small addition to one search interval in exchange for a good bound on the next search interval. 

As a final remark, the sums in (\ref{eq:sum_norm}) and (\ref{eq:sum_disc}) cannot be computed exactly because there is no known polynomial time algorithm for factoring integers. In line of Subroutine~\ref{sub:2}, we define a bound $r$ beyond which the primes contribute negligibly. 

\subsection{The algorithm} We use the following standard notation:

\begin{definition}\label{def:radical}The \emph{radical} of an ideal $\fa$ in a commutative ring $\cO$ is \[\sqrt{\fa}\coloneqq \{\alpha\in\cO:\alpha^k\in\fa\text{ for some }k\geq 0\}.\]\end{definition} An alternative phrasing of $2\in \sqrt{\fa}$ for some finite index ideal $\fa\subseteq\cO$ (as in lines 6 and 10 below, where $\cO=\bZ[x]$) is $[\cO:\fa]=2^k$ for some $k\geq 0$.

\begin{definition}\label{def:valuation}Let $\cO$ be a commutative ring. The \emph{valuation} of an ideal $\fa\subseteq\cO$ at the prime $\fp\subseteq\cO$, denoted $v_{\fp}(\fa)$, is the largest integer $k$ such that $\fa\subseteq\fp^k$.\end{definition}

\begin{algorithm}
\SetAlgorithmName{Subroutine}{Subroutine}{}
    \caption{Deterministic method to adjust $A$ so that $\det(xI_n-A)$ is irreducible and the output ideal of Algorithm~\ref{alg:1} is invertible.}\label{sub:3}
    \DontPrintSemicolon
    \KwIn{A symmetric $A\in\Mat_n(\bZ)$ and a real number $\delta>1$}
    \KwOut{A scaled, perturbed copy of $A$}
    $q'\gets n^2q\lceil \|A^{-1}\|_2(\sqrt{\kappa}+1)/(\sqrt{\kappa}-1)\rceil$\commentR{$\triangleright\;$see Definition \ref{def:Sq}}
    $A\gets 3q'\lceil\log(q'\|A\|_2)\rceil A+S_q$\;
    $r\gets n^2(\log\|A\|_2)$\;
    \For{$j$\textup{ from }$2$\textup{ to }$n$}{
    $s\gets\sum_p\smash{\frac{1}{p}}v_p([\bZ[x]:\fa_j])$ for $4n^2< p<r$\commentR{$\triangleright\;$see Definitions \ref{def:Aj} and \ref{def:valuation}}
    \While(\commentF{$\triangleright\;$see Definition \ref{def:radical}}){$2\not\in\!\smash{\sqrt{(\smash{\fa_j},\smash{f'_{j-1}}(x))}}$\,\textup{ or }$s>\frac{j}{3n}$}{
    $A_{j,j-1},A_{j-1,j}\gets A_{j,j-1}+q$\;
    recompute $\fa_j$ and $s$
    }
    $s\gets\sum_p\smash{\frac{1}{p}}v_p(\text{disc}f_j(x))$ for $4n^2< p< r$\;
    \While{$2\not\in\!\smash{\sqrt{(\smash{\fa_j},\smash{f_j'}(x))}}$\,\textup{ or }$s>\frac{j}{3n}$}{
    $A_{j,j}\gets A_{j,j}+q$\;
    recompute $f_j(x)$ and $s$
    }
    }
    \Return{$A,r$}
\end{algorithm}

The following lemma is known, in that it is an immediate consequence of established results. A proof is included for completeness.

\begin{lemma}\label{lem:1/p^2}For all $x\geq 2$, we have for primes $p$ that \[\sum_{p> x}\frac{1}{p^2}<\frac{1}{x\log x}-\frac{1}{3x\log^2x}.\]\end{lemma}

\begin{proof}We evaluate the Riemann--Stieltjes integral with integration by parts: \[\sum_{p>x}\frac{1}{p^2}=\int_x^{\infty}\frac{1}{t^2}d\pi(t)=-\frac{\pi(x)}{x^2}+\int_{x}^{\infty}\frac{2\pi(t)}{t^3}dt.\] (A weaker version of) Dusart's inequality for $x\geq 599$, \cite[Corollary 5.2]{dusart}, \[\frac{x}{\log x}\left(1+\frac{1}{\log x}\right)<\pi(x)<\frac{x}{\log x}\left(1+\frac{4}{3\log x}\right),\] then bounds our initial sum from above. We apply integration by parts twice to the result: \begin{align*}-\frac{1}{x\log x}-\frac{1}{x\log^2 x}+\int_x^{\infty}&\!\left(\frac{2}{t^2\log t}+\frac{8}{3t^2\log^2 t}\right)\!dt\\ &=\frac{1}{x\log x}-\frac{1}{3x\log^2x}-\int_x^{\infty}\!\frac{4}{3t^2\log^3 t}dt.\end{align*} The final integral is positive, so the proof is complete for $x\geq 599$. The bound is easily checked for $2\leq x<599$.\end{proof}

\begin{lemma}\label{lem:ajbj}Let $r$ be as in line 3 and recall notation from (\ref{eq:ajbj}). The perturbations $a_j$ and $b_j$ found in the \textup{\textbf{while}} loops of Subroutine \ref{sub:3} are bounded by $qr$.\end{lemma}

\begin{proof} Let $t$ denote the value of $\|A\|_2$ used in line 3. From Definition \ref{def:Sq}, $q=e^{\vartheta(4n^2)}$, where $\vartheta$ is Chebyshev's theta function. (A weaker version of) Rosser and Schoenfield's lower bound \cite[Corollary (3.16)]{rosser} gives $e^{\vartheta(4n^2)} > e^{3n^2}$ whenever $n\geq 4$. Combining this with lines 1 and 2 crudely gives \begin{equation}\label{eq:tbound}t\geq 3q'\log(q')-\|S_q\|_2\geq 3n^2q\log(n^2q)-\frac{nq}{2}>9n^4e^{3n^2}.\end{equation} whenever $n\geq 4$. That this lower bound on $t$ also holds when $n=2$ or $3$ is easily checked. The penultimate expression above also shows (crudely again) that \begin{equation}\label{eq:qbound}q<\frac{t}{2n^2\log t}.\end{equation}

To simplify the exposition, we first bound $|\text{disc} f_j(x)|$ assuming $0\leq a_i,b_i\leq qr$ for all $i\leq j$. The spectral norm of any minor is bounded above by that of the matrix, and there are at most two values of $a_i$ and $b_i$ in any row or column (see (\ref{eq:ajbj})), so all minors have spectral norm at most $t+2qr$. A minor's determinant of $A_j$ is bounded by its spectral norms raised to the power of its dimensions. Thus the coefficient of $x^{j-i}$ in $f_j(x)$ is bounded by $\binom{j}{i}(t+2qr)^i$, and the coefficient of $x^{j-1-i}$ in $f_j'(x)$ is bounded by $\binom{j}{i}(j-i)(t+2qr)^i<\binom{j}{i}2^{-3j}(t+2qr)^{i+1}$. Hence the sum of all coefficient magnitudes of $f_j((t+2qr)x)$ is bounded by \[\sum_{i=0}^j\binom{j}{i}(t+2qr)^i(t+2qr)^{j-i}=2^j(t+2qr)^j.\] Similarly, the sum of coefficient magnitudes of $f'_j((t+2qr)x)$ is bounded by $2^{-2j}(t+2qr)^j$. These sums bound the $\ell_2$ lengths of the columns of the Sylvester matrix for $f_j((t+2qr)x)$ and $f'_j((t+2qr)x)$, whose determinant is bounded by Hadamard's inequality in the third line below: \begin{align*}|\text{disc}f_j(x)|&=|\text{res}(f_j(x),f_j'(x))|\\&=\frac{1}{(t+2qr)^{j(j-1)}}\left|\text{res}\left(f_j((t+2qr)x),f_j'((t+2qr)x)\right)\right|\\&\leq \frac{2^{j(j-1)-2j(j)}(t+2qr)^{j(j+j-1)}}{(t+2qr)^{j(j-1)}}\\&< (t+qr)^{j^2}.\end{align*} Now, from (\ref{eq:qbound}) and the definition of $r$, $t+2qr$ is bounded by $t(1+1/\log n)$. In particular, if $j<n$ we have \begin{equation}\label{eq:discbound}\frac{\log|\text{disc}f_j(x)|}{r\log n}<\frac{(n-1)^2\log(t(1+1/\log n))}{n^2\log t\log n}.\end{equation} This is a decreasing function of $t$. Substituting the lower bound on $t$ from  (\ref{eq:tbound}) into (\ref{eq:discbound}) gives a function of $n$ that is less than $\frac{4}{9}$ for all $n\geq 2$. Therefore, \begin{equation}\label{eq:rlogn}\log|\text{disc}f_j(x)|<\frac{4}{9}r\log n.\end{equation}

Now fix some \textbf{for} loop iteration $2\leq j\leq n$, and assume for induction that $a_i,b_i\leq qr$ for all $i<j$. Consider the \textbf{while} loop in line 6. Forcing $2\in\sqrt{(\smash{\fa_j,f_{j-1}'(x)})}$ is equivalent to $(\fa_j,f_{j-1}'(x),p)=\bZ[x]$ for every odd prime $p$. If $p<4n^2$, then $\fa_j$ contains the element $g_{j-1}(x)$ defined in Definition \ref{def:Aj}. Since $A\equiv S_p\,\text{mod}\,p$ by line~2, we see from Definition \ref{def:Sq} that the bottom-left $(j-1)\times(j-1)$ minor determinant of $xI_j-A_j$, called $g_{j-1}(x)$ as per Definition~\ref{def:Aj}, is equivalent to $1\,\text{mod}\,p$. Hence $(\fa_j,f_{j-1}'(x),p)=\bZ[x]$ is automatic when $p<4n^2$. 

Suppose $p > 4n^2$. If $j\geq 3$, line 10 from the previous \textbf{for} loop iteration guarantees \begin{equation}\label{eq:j-1coprime}(\fa_{j-1},f_{j-1}'(x),p)=\bZ[x].\end{equation} We also have \begin{align}\label{eq:4/9}\nonumber&&\sum_{p> 4n^2}\frac{1}{p}v_p(\text{disc} f_{j-1}(x))& =\sum_{4n^2< p<r}\frac{1}{p}v_p(\text{disc} f_{j-1}(x))+\sum_{p\geq r}\frac{1}{p}v_p(\text{disc} f_{j-1}(x))\hspace{-3cm}&&\\\nonumber && &\leq \frac{j-1}{3n}+\sum_{p\geq r}\frac{1}{p}v_p(\text{disc} f_{j-1}(x)) && \text{by line 10 for }j-1\\\nonumber&& &< \frac{1}{3}+\frac{1}{r\log r}\sum_{p\geq r}v_p(\text{disc} f_{j-1}(x))\log p \hspace{-1.9cm} && \\\nonumber && &\leq \frac{1}{3}+\frac{\log|\text{disc}f_{j-1}(x)|}{r\log r} && \\\nonumber && &< \frac{1}{3}+\frac{\log|\text{disc}f_{j-1}(x)|}{r\log (n^2\log e^{3n^2})} && \text{by (\ref{eq:tbound})}\\ \nonumber&& &<\frac{1}{3}+\frac{\log|\text{disc}f_{j-1}(x)|}{4 r\log n} && \\ && & < \frac{4}{9} && \text{by (\ref{eq:rlogn}).}\end{align} Properties (\ref{eq:j-1coprime}) and (\ref{eq:4/9}) also hold when $j=2$. Indeed, $f_1(x)=x-A_{1,1}$, so $(\fa_1,f_1'(x),p)=(\fa_1,1,p)=\bZ[x]$ and $\sum_p \frac{1}{p}v_p(\text{disc}f(x))=\sum_p \frac{1}{p}v_p(1)=0$.

Let $\fp$ be the pullback to $\bZ[x]$ of a ramified in $\bZ[x]/(f_{j-1}(x))$ above $p$, meaning $\fp$ is generated by $p$ and a common irreducible factor of $f_{j-1}(x)\,\text{mod}\,p$ and $f'_{j-1}(x)\,\text{mod}\,p$. As we checked in (\ref{eq:j-1coprime}), such a prime does not contain $\fa_{j-1}$. By computing the first $j-1$ generators of $\fa_j$ as in (\ref{eq:cramer}), we see that $\fa_{j-1}\not\subseteq\fp$ implies there is at most one congruence class $a_j\,\text{mod}\,p$ that makes $\fa_j\subseteq\fp$. Since $v_p(\text{disc}f_{j-1}(x))$ is an upper bound on the number of ramified primes above $p$, we conclude that the total number of congruence classes mod$\,p$ that $a_j$ must avoid, call this number $w(p)$, is at most $v_p(\text{disc}f_{j-1}(x))$. 

As $a_j$ increments through the sequence $0,q,\dots,\lfloor r\rfloor q$, it lands on one of the $w(p)$ congruence classes mod$\,p$ at most $w(p)\big(\frac{r}{p}+1\big)$ times. So altogether, the total number of values of $a_j$ in this sequence of length $\lfloor r\rfloor+1$ for which $(\fa_j,f_{j-1}'(x))$ is contained in a prime above some $p> 4n^2$ is at most \begin{align}\label{eq:2/3r-1}\nonumber&& \sum_{p> 4n^2}w(p)\!\left(\frac{r}{p}+1\right)&\leq r\!\!\sum_{p> 4n^2}\!\frac{v_p(\text{disc}f_{j-1}(x))}{p}+\sum_{p> 4n^2}\!v_p(\text{disc}f_{j-1}(x))&&\\\nonumber&& &\leq \frac{4r}{9}+\sum_{p>4n^2}\!v_p(\text{disc}f_{j-1}(x)) && \text{by (\ref{eq:4/9})}\\\nonumber && &<\frac{4r}{9}+\frac{1}{\log 4n^2}\sum_{p> 4n^2}\!v_p(\text{disc}f_{j-1}(x))\log p&&\\\nonumber && &< \frac{4r}{9}+\frac{\log|\text{disc}f_{j-1}(x)|}{2\log n} &&\\ && &<\frac{2r}{3} && \text{by (\ref{eq:discbound}).}\end{align}

Next we count how many choices of $a_j$ among $0,q,\dots,\lfloor r\rfloor q$ fail line 6 because they make $\sum_p\frac{1}{p}v_p([\bZ[x]:\fa_j])$ too large for primes $4n^2< p < r$. Since $\fa_j$ contains $f_{j-1}(x)$ for any choice of $a_j$, we only ever consider primes $\fp\subseteq\bZ[x]$ generated by $p>4n^2$ together with some irreducible factor of $f_{j-1}(x)\,\text{mod}\,p$. Given $\fp$, let $w_{\fp}$ denote the product of the inertial degree, call it $d(\fp|p)$ (the degree of the corresponding irreducible factor of $f_{j-1}(x)\,\text{mod}\,p$), with the minimal valuation of $\fa_j$ at $\fp$ as $a_j$ runs through the sequence $0,q,\dots,\lfloor r\rfloor q$. Also let $w'_{\fp}(a)$ denote the value of $d(\fp|p)v_{\fp}(\fa_j)-w_{\fp}$ when $\fa_j$ is determined by $a_j=aq$. This way, \begin{equation}\label{eq:w_p(a)}\sum_{\fp\,|\,p}(w_{\fp}+w'_{\fp}(a))=v_p([\bZ[x]:\fa_j]).\end{equation} 

The coefficients of $a_j$ in (\ref{eq:cramer}) all lie in $\fa_{j-1}$, and thus they lie in $\fp^k$ with $k=v_{\fp}(\fa_{j-1})$. In particular, changing $a$ has no effect on the generators of $\fa_j$ reduced mod$\,\fp^k$. However, since $\fa_{j-1}\not\subseteq\fp^{k+1}$, there is at most one choice of $a_j\,\text{mod}\,p$ for which $v_{\fp}(\fa_j)$ could possibly exceed $k$. In other words, $w_{\fp}\leq d(\fp|p)k$. Thus for any choice of $a_j$, \begin{align}\label{eq:a_j-1bound}\nonumber && \sum_{4n^2<p<r}\sum_{\fp\,|\,p}\frac{w_{\fp}}{p}&\leq \sum_{4n^2<p<r}\sum_{\fp\,|\,p}\frac{1}{p}d(\fp|p)v_{\fp}(\fa_{j-1}) && \\ \nonumber && &=\sum_{4n^2<p<r}\!\frac{1}{p}v_p([\bZ[x]:\fa_{j-1}]) && \\ && &\leq\frac{j-1}{3n} && \text{by line 6 for }j-1.\end{align} Furthermore, we know $w'_{\fp}(a)$ is positive at most once for $a=0,\dots,p-1$, and its value cannot exceed $d(\fp|p)v_{\fp}(f_{j-1}(x))$ because $f_{j-1}(x)\in\fa_j$. But this product is simply the degree of the largest power of the irreducible factor defining $\fp$ that divides $f_{j-1}(x)\,\text{mod}\,p$. Since $f_{j-1}(x)$ is monic of degree $j-1$, we have \begin{equation}\label{eq:w_pbound}\sum_{\fp\,|\,p}\sum_{a=0}^{p-1}w'_{\fp}(a)\leq j-1.\end{equation} The number of times $a$ hits any given congruence class mod$\,p$ in the sequence $0,1,\dots,\lfloor r\rfloor$ is at most $\frac{r}{p}+1$, allowing us to bound the average of the double sum denoted $S(a)$ below: \begin{align}\label{eq:badsum}\nonumber && \frac{1}{\lfloor r\rfloor+1}\sum_{a=0}^{\lfloor r\rfloor}\!\smash{\underbrace{\sum_{4n^2<p<r}\sum_{\fp\,|\,p}\frac{w'_{\fp}(a)}{p}}_{S(a)}}&<\sum_{4n^2<p<r}\!\frac{1}{rp}\!\left(\frac{r}{p}+1\right)\!\sum_{\fp\,|\,p}\sum_{a=0}^{p-1}w'_{\fp}(a)\hspace{-1cm} && \\\nonumber && &\leq (j-1)\!\!\!\sum_{4n^2<p<r}\!\!\left(\frac{1}{p^2}+\frac{1}{r p}\right) && \text{by (\ref{eq:w_pbound})}\\&& &<(j-1)\!\left(\sum_{p>4n^2}\frac{1}{p^2}+\frac{1}{r}\!\sum_{4n^2<p<r}\frac{1}{p}\right). && \end{align} Setting $x=4n^2$ in the statement of Lemma~\ref{lem:1/p^2} bounds the first summation above, and the second is bounded according to Rosser and Schoenfeld by $\log\log r-1+(\log r)^{-2}$ \cite[Corollary (3.20)]{rosser}. (Note that the constant $-1$ in our bound is less than in \cite[Corollary (3.20)]{rosser}; this is more than accounted for by our lack of primes less than $4n^2$.) Using $r=n^2\log t \geq 3n^4$ by (\ref{eq:tbound}), we see that the average value of $S(a)$ is at most \begin{equation}\label{eq:finalp}(n-1)\left(\frac{1}{4n^2\log 4n^2}-\frac{1}{12n^2\log^2 4n^2}+\frac{\log\log 3n^4-1+(\log 3n^4)^{-2}}{3n^4}\right).\end{equation} This is less than $\frac{1}{9n}$ for all $n\geq 2$.

No more than one-third of all $a$ can make $S(a)$ at least thrice the average. For the remaining choices of $a$, we have $S(a)<\frac{1}{3n}$, giving \begin{align*}&&\sum_{4n^2<p<r}\frac{1}{p}v_p([\bZ[x]:\fa_j])&=\sum_{4n^2<p<r}\sum_{\fp\,|\,p}\frac{1}{p}(w_{\fp}+w'_{\fp}(a))&& \text{by (\ref{eq:w_p(a)})}\\&& &\leq \frac{j-1}{3n}+S(a)&&\text{by (\ref{eq:a_j-1bound})}\\&& &<\frac{j}{3n}, && \end{align*} just as line 6 demands.  

Combining this with (\ref{eq:2/3r-1}) shows that the total number of failures among $a_j=0,q,\dots,\lfloor r\rfloor q$ for either condition in line 6 is less than $\frac{r}{3}+\frac{2r}{3}<\lfloor r\rfloor+1$. In particular, at least one success is guaranteed, so $a_j<qr$.

The proof that $b_j<qr$ is nearly identical, but with $[\bZ[x]:\fa_j]$ now playing the role of $\text{disc}f_{j-1}(x)$, and $\text{disc}f_j(x)$ now playing the role of $[\bZ[x]:\fa_j]$. Let us highlight the minor adjustments.

To start, (\ref{eq:rlogn}) can be deduced for $[\bZ[x]:\fa_{j+1}]$ in place of $\text{disc}f_j(x)$ by bounding the resultant of $f_j(x)$ and $g_j(x)$ rather than $f_j(x)$ and $f'_j(x)$. This is because \[|\text{res}(f_j(x),g_j(x))|=[\bZ[x]:(f_j(x),g_j(x))]\leq[\bZ[x]:\fa_{j+1}].\] (The first equality requires $\text{res}(f_j(x),g_j(x))\neq 0$, which holds because $f_j(x)$ is monic and $g_j(x)\equiv1\,\text{mod}\,p$ for any $p<4n^2$.) The coefficient of $x^{j-1-i}$ in $g_j(x)$ is bounded using Hadamard's inequality by $\binom{j-1}{i}(\sqrt{i+1})^{i+1}\|A_{j+1}\|^{i+1}<(j\|A_{j+1}\|)^{i+1}$. Thus all coefficients of $g_j(j\|A_{j+1}\|x)$ are bounded by $(j\|A_{j+1}\|)^j$, the same bound that we used for coefficients of $f'_j(x)$, but with $\|A_{j+1}\|$ in place of $\|A_j\|$ (which does not matter because both are replaced with $t+qr$).

Reproducing (\ref{eq:rlogn}) allows us to deduce (\ref{eq:4/9}) for $\sum_p\frac{1}{p}v_p([\bZ[x]:\fa_j])$. Fixing some $p>4n^2$, we just chose $a_j$ so that $\fa_{j-1}$ in (\ref{eq:j-1coprime}) can be replaced with $\fa_j$. This means $-f'_{j-1}(x)$, which is coefficient of $b_j$ in $f'_j(x)$, does not belong to any $\fp\,|\,p$ that contains $\fa_j$. Thus there is at most one congruence class $b_j\,\text{mod}\,p$ per prime $\fp\,|\,p$ that violates the condition $(\fa_j,f'_j(x),p)=\bZ[x]$ implicitly stated in line 10. There are at most $v_p([\bZ[x]:\fa_j])$ such primes $\fp$. As a result, all of the arithmetic in (\ref{eq:2/3r-1}) is still correct with $[\bZ[x]:\fa_j]$ in place of $\text{disc}f_{j-1}(x)$. 

The most significant change comes occurs in (\ref{eq:w_pbound}) in the effort to avoid those $b_j$ that make $\sum_p\frac{1}{p}v_p(\text{disc}f_j(x))$ too large. Let $\fb_j=(f_j(x),f'_j(x))$. Since $|\text{disc}f_j(x)|=|\text{res}(f_j(x),f'_j(x))|=[\bZ[x]:\fb_j]$, we define $w_{\fp}$ as the minimum of $d(\fp|p)v_{\fp}(\fb_j)$ as $b_j$ ranges over $0,q,\dots,\lfloor r\rfloor q$. Also define $w_{\fp}(b)$ as the value of $d(\fp|p)v_{\fp}(\fb_j)-w_{\fp}$ when $f_j(x)$ is determined by $b_j=bq$. Unlike the situation with $a_j$ and $\fa_j$, both generators of $\fb_j$ are affected by $b_j$. Indeed, $f_j(x)=(x-A_{j,j}-b_j)f_{j-1}(x)+h(x)$ and $f_j'(x)=f_{j-1}(x)+(x-A_{j,j}-b_j)f'_{j-1}(x)+h'(x)$ for some $h(x)$ of degree at most $j-2$. This shows two things: First, $w_{\fp}$ is bounded by $d(\fp|p)v_{\fp}(\fb_{j-1})$, where $\fb_{j-1}$ is the ideal generated by the two coefficients of $b_j$, namely $f_{j-1}(x)$ and $f'_{j-1}(x)$. Therefore, \[\sum_{\fp\,|\,p}w_{\fp}\leq v_p([\bZ[x]:\fb_{j-1}])=v_p(\text{disc}f_{j-1}(x)),\] allowing (\ref{eq:a_j-1bound}) to be replicated. Second, the only primes $\fp$ that can possible contain $\fb_j$ are those generated by $p$ and an irreducible factor of \[f_{j-1}(x)(f_{j-1}(x)+(x-A_{j,j})f'_{j-1}(x)+h'(x))-f'_{j-1}(x)((x-A_{j,j})f_{j-1}(x)+h(x))\] \[=f_{j-1}^2(x)+f_{j-1}(x)h'(x)-f'_{j-1}(x)h(x),\] which is monic of degree $2j-2$. Hence the upper bound in (\ref{eq:w_pbound}) gets doubled, thereby doubling (\ref{eq:finalp}). Nevertheless, (\ref{eq:finalp}) is still less than $\frac{1}{9n}$ for all $n\geq 5$. For $n=2$, $3$, and $4$, (\ref{eq:badsum}) can be bounded by $\frac{1}{9n}$ by using sharper bounds on $\sum_p\frac{1}{p^2}$ than Lemma~\ref{lem:1/p^2} provides, namely $0.017$, $0.0062$, and $0.0034$, respectively.\end{proof}

The number of \textbf{while} loop iterations in Algorithm~\ref{alg:1} is handled similarly:

\begin{lemma}\label{lem:beta}With $r$ as defined in Subroutine~\ref{sub:3}, the \textup{\textbf{while}} loop in Algorithm~\ref{alg:1} terminates after at most $r$ iterations.\end{lemma}

\begin{proof}Let $f(x)=f_n(x)$ be the minimal polynomial of $\lambda$, let $\bb$ assume its value in line 11 of Algorithm~\ref{alg:1}, and let $b(x)=(1,x,\dots,x^{n-1})\bb$. Observe that $\beta$ is the image of $b(x)$ under $\bZ[x]\to\bZ[x]/(f(x))\xrightarrow{\sim}\bZ[\lambda]$. Since the conductor $\ff(\bZ[\lambda])$ is contained in the ideal generated by $f'(\lambda)$, we are done if we can show that $(b(x)+a,f(x),f'(x))=\bZ[x]$ for some $a=0,\dots,\lfloor r\rfloor$.

For each prime $\fp\,|\,p$ in $\bZ[x]$ that divides $(f(x),f'(x))$, $a$ only needs to avoid at most one congruence class mod$\,p$. Every such prime $\fp$ contributes at least one to $v_p(\text{disc}f(x))$, which we know to be zero when $p<4n^2$ by Lemma~\ref{lem:disc}. Replacing $n-1$ with $n$ in the numerator of (\ref{eq:discbound}) shows that \begin{equation}\label{eq:newdisc}\frac{\log|\text{disc} f(x)|}{r\log n}<\frac{8}{9}\end{equation} when $n\geq 4$. The same bound holds when $n=2$ or $3$, but we have to be slightly more precise to see it: The lower bound on $t$ is $3q'\log q'$, where $q'\geq n^2q$. This gives $t> 4.32\cdot 10^6$ when $n=2$ and $1.57\cdot 10^{14}$ when $n=3$. We use $\binom{n}{i}(t+2qr)^i$ again to bound the coefficient of $x^i$ in $f(x)$, but now we simply plug that expression directly into the formulas for the discriminant of quadratic and cubic polynomials. The result is $|\text{disc}f(x)|<8(t+2qr)^2$ when $n=2$ and $|\text{disc}f(x)|<486(t+2qr)^6$ when $n=3$. Plugging in the exact values of $q$, setting $r=n^2\log t$, and evaluating at our lower bound for $t$ verifies (\ref{eq:newdisc}) for $n=2$ and $3$. 

The exact same chain of inequalities in (\ref{eq:4/9}), but with (\ref{eq:newdisc}) used in the final step, ends with an upper bound of $\frac{1}{3}+\frac{1}{4}(\frac{8}{9})=\frac{5}{9}$ instead of $\frac{4}{9}$. Finally, replicating the inequalities in (\ref{eq:2/3r-1}), but with $\frac{5r}{9}$ used instead of $\frac{4r}{9}$ in the second step, ends with an upper bound of $r$. Thus there is at least one value of $a$ in the sequence $0,\dots,\lfloor r\rfloor$ that passes the \textbf{while} loop condition.\end{proof}

\begin{theorem}\label{thm:time3}Algorithm~\ref{alg:1} with Subroutine~\ref{sub:3} runs in deterministic polynomial time.\end{theorem}

\begin{proof}By Lemmas \ref{lem:ajbj} and \ref{lem:beta}, all \textbf{while} loops terminate in a polynomial number of iterations. Except perhaps line 13 in Algorithm~\ref{alg:1} and lines 6 and 10 in Subroutine \ref{sub:3}, all steps rely on well-known algorithms that run in polynomial time in their inputs, which have polynomial bit length in the bit length of $M$ and $\kappa$. Methods for line 13 and line 6 and 10 are also established, though perhaps not as well known outside of algorithmic number theory. Ideal or radical ideal membership can be tested in the ring $\bZ[x]/(f(x))$ by scaling first computing a $\bZ$-basis for the ideal. (Scale generators over the ring by $1,x,\dots,x^{n-1}$ and reduce the results mod$\,f(x)$ to produce generators over $\bZ$, then reduce to a basis.) Once the ideal becomes a sublattice of $\bZ$, we can check the smallest multiple of $1+0x+\cdots+0x^{n-1}$ that it contains. If that multiple is 1, line 13 in Algorithm \ref{alg:1} passes the test; if that multiple is a power of $2$, line 6 or 10 in Subroutine \ref{sub:3} passes the test. \end{proof}

\begin{theorem}\label{thm:sub3}From Algorithm~\ref{alg:1} with Subroutine~\ref{sub:3}, $\beta\alpha_1,\dots,\beta\alpha_n$ is a $\bZ$-basis for an ideal $\beta\fa$ that is coprime to the conductor in the totally real number ring $\bZ[\lambda]$. Furthermore, if $B$ is the matrix with $i^\text{th}$ column $\Sigma(\alpha_i\beta)$ and $A$ is its original value in line 2, then $AB^{-1}\Sigma(\beta\fa)=M\bZ^n$, and $\kappa_2(AB^{-1})<\kappa$.\end{theorem}

\begin{proof}As in the proof of Theorem~\ref{thm:sub2}, let $A_1=A$ and $B_1$ denote the input and output of Subroutine~\ref{sub:3}, and let $A_2=B_1$ and $B_2=B$.

First observe that $\bQ(\lambda)$ is a totally real number field of degree $n$. Indeed, all perturbations found in the \textbf{for} loop of Subroutine~\ref{sub:3} are even and preserve symmetry. Thus $B_1\equiv S_2\,\text{mod}\,2$ by line 2 of Subroutine~\ref{sub:3}. From Lemma~\ref{lem:S2}, the characteristic polynomial of $S_2$ is irreducible mod$\,2$. Hence the characteristic polynomial of $A_2=B_1$, which is the minimal polynomial of $\lambda$ is irreducible with only real roots.

Next, coprimality of $\beta$ and $\ff(\bZ[\lambda])$ is ensured by the \textbf{while} loop in Algorithm~\ref{alg:1}. Line 10 in Subroutine~\ref{sub:3} shows that $(\fa,f'(\lambda))$ is not contained in any prime ideals above some $p>4n^2$, which implies the same for $(\fa,\ff(\bZ[\lambda]))$. For primes $p<4n^2$, we have Lemma~\ref{lem:disc}.

Finally, we check that $\kappa_2(A_iB_i^{-1})<\sqrt{\kappa}$ for $i=1,2$. The proof for $A_2B_2^{-1}$ is almost identical to that in the proof of Theorem~\ref{thm:sub2}. Let us use the same notation for the diagonal matrices $C$ and $C'$. The only difference now is that the initial upper bound of $1$ in line 10 of Algorithm~\ref{alg:1} can grow throughout the \textbf{while} loop in lines 12--13, but it is capped at $\lfloor r\rfloor$ by Lemma~\ref{lem:beta}. This is perfectly accounted for by the presence of $r$ on the right side of line 7 in Algorithm~\ref{alg:1}. So the analysis in (\ref{eq:CC'}) holds, and we obtain $\kappa_2(A_2B_2^{-1})<\sqrt{\kappa}$. 

To deal with $A_1B_1^{-1}$, we need an analog of (\ref{eq:A1B1}). Let $t$ denote the value of ``$\|A\|_2$'' in line 3 of Subroutine~\ref{sub:3}, just as in the proof of Lemma~\ref{lem:ajbj}. By line 2, \[t\geq 3q'\|A_1\|_2\log(q'\|A_1\|_2)-\|S_q\|\geq 3q'\|A_1\|_2\log(q'\|A_1\|_2)-\frac{q'}{2n},\] which implies \begin{equation}\label{eq:tlogt}\frac{t}{\log t} > \frac{5}{2}q'\|A_1\|_2.\end{equation}

The difference between $3q'\lceil\log (q'\|A_1\|_2)\rceil A_1$ and $B_1$ consists of entries from $S_q$ and $a_j$ and $b_j$. Lemma~\ref{lem:ajbj} bounds $a_j$ and $b_j$ by $qr$, and the entries of $S_q$ are bounded in magnitude by $\frac{q}{2}$. There are at most two perturbations $a_j$ and $b_j$ in each row and column of $B_1$. Since the $1$-norm and $\infty$-norm of a matrix are bounded by the maximum sum of entry magnitudes among its columns and rows, we have  \begin{align*}&& \left\|A_1-\frac{B_1}{3q'\lceil\log (q'\|A\|_2)\rceil}\right\|_2&\leq\max_{p=1,\infty}\left\|A_1-\frac{B_1}{3q'\lceil\log (q'\|A\|_2)\rceil}\right\|_p\hspace{-1.5cm} && \\ && &\leq\frac{2qr+nq/2}{3q'\lceil\log (q'\|A\|_2)\rceil} && \\ && &<\frac{5\|A_1\|_2qr}{2(3q'\|A_1\|_2\lceil\log (q'\|A_1\|_2)\rceil+nq/2)} \hspace{-1.5cm}&&\\ && &\leq \frac{5\|A_1\|_2qn^2\log t}{2t} && \text{by definition of }r\text{ and }t\\ && &<\frac{qn^2}{q'}&& \text{by (\ref{eq:tlogt})}\\ & && \leq \frac{\sqrt{\kappa}-1}{\|A_1^{-1}\|_2(\sqrt{\kappa}+1)} && \text{by definition of }q'.\end{align*} It follows that $\kappa_2(A_1B_1^{-1})<\sqrt{\kappa}$ by Lemma~\ref{lem:condition}.\end{proof}

This justifies Theorem~\ref{thm:main} in the introduction. We can also now prove the claimed hardness result from the introduction.

\begin{corollary}Restricted to the canonical embedding of ideals coprime to the conductor in totally real, monogenic number rings, $\SVP$ and $\CVP$ are \NP-hard to approximate within factors of $\sqrt{2}$ and $\smash{n^{\frac{1}{2}-\varepsilon}}$, respectively.\end{corollary}

\begin{proof}This is almost immediate from Theorems \ref{thm:time3} and \ref{thm:sub3}. It remains only to check that the value of $\kappa$ needed to preserve the factors $\sqrt{2}$ and $\smash{n^{c/\log\log n}}$ is not so small that it destroys running time. 

In \cite{wan}, Wan proves that $\sqrt{2}$-$\SVP$ is \NP-hard in the setting of integer lattices. In other words, Proposition~\ref{prop:exact} applies---we may perfectly preserve the factor $\sqrt{2}$ by making $\kappa$ sufficiently small in Algorithm~\ref{alg:1}. Note that the bound Proposition~\ref{prop:exact} places on $\kappa$ gives $\lceil (\sqrt{\kappa}-1)/(\sqrt{\kappa}-1)\rceil$ bit length that is polynomial in the bit length of $A$ from line 2 of Algorithm~\ref{alg:1}, which is polynomial in the bit length of $M$ using Storjohann's Smith Normal Form algorithm \cite{storjohann}.

We need not bother with Proposition~\ref{prop:exact} to prove \NP-hardness of $\CVP$. Song's factor of $\smash{n^{\frac{1}{2}-\varepsilon}}$ is not fixed like $\sqrt{2}$. For example, $\kappa=2$ works. \end{proof}

\printbibliography

@inproceedings{dwork,
  title={A public-key cryptosystem with worst-case/ average-case equivalence},
  author={Ajtai, Mikl{\'o}s and Dwork, Cynthia},
  booktitle={Proceedings of the Twenty-Ninth Annual {ACM} Symposium on Theory of Computing},
  pages={284--293},
  year={1997},
  doi={10.1145/258533.258604}
}

@inproceedings{hoffstein,
  title={{NTRU}: A ring-based public key cryptosystem},
  author={Hoffstein, Jeffrey and Pipher, Jill and Silverman, Joseph H.},
  booktitle={International Algorithmic Number Theory Symposium ({ANTS-III})},
  pages={267--288},
  year={1998},
  publisher={Springer},
  volume={1423},
  series={Lecture Notes in Computer Science},
  doi={10.1007/BFb0054868}
}

@inproceedings{ajtai,
  title={Generating hard instances of lattice problems},
  author={Ajtai, Mikl{\'o}s},
  booktitle={Proceedings of the Twen\-ty-Eighth Annual {ACM} Symposium on Theory of Computing},
  pages={99--108},
  year={1996},
  doi={10.1145/237814.237838}
}

@inproceedings{regev,
  author    = {Regev, Oded},
  title     = {On lattices, learning with errors, random linear codes, and cryptography},
  booktitle = {Proceedings of the Thirty-Seventh Annual ACM Symposium on Theory of Computing},
  year      = {2005},
  publisher = {ACM},
  doi       = {10.1145/1060590.1060603}
}

@inproceedings{peikert,
  author  = {Peikert, Chris},
  title   = {Public-key cryptosystems from the worst-case shortest vector problem},
  booktitle = {Proceedings of the Forty-First Annual ACM Symposium on Theory of Computing},
  year    = {2009},
  doi     = {10.1145/1536414.1536461}
}

@article{gentry,
  author  = {Gentry, Craig},
  title   = {Fully homomorphic encryption using ideal lattices},
  journal = {Proceedings of the Forty-First Annual ACM Symposium on Theory of Computing},
  year    = {2009},
  pages   = {169--178},
  doi     = {10.1145/1536414.1536440}
}

@inproceedings{brakerski,
  author    = {Brakerski, Zvika and Vaikuntanathan, Vinod},
  title     = {Efficient Fully Homomorphic Encryption from (Standard) {LWE}},
  booktitle = {Proceedings of the IEEE 52nd Annual Symposium on Foundations of Computer Science ({FOCS})},
  year      = {2011},
  pages     = {97--106},
  publisher = {IEEE},
  doi       = {10.1109/FOCS.2011.12}
}

@article{lyub,
  author  = {Lyubashevsky, Vadim and Micciancio, Daniele},
  title   = {Generalized compact knapsacks are collision resistant},
  journal = {Lecture Notes in Computer Science},
  year    = {2006},
  pages   = {144--155},
  doi     = {10.1007/11787006_13}
}

@article{micc2,
  author  = {Micciancio, Daniele},
  title   = {Generalized compact knapsacks, cyclic lattices, and efficient one-way functions},
  journal = {Computational Complexity},
  year    = {2007},
  volume  = {16},
  pages   = {365--411},
  doi     = {10.1007/s00037-007-0234-9}
}

@misc{campbell,
  title={{Soliloquy}: A cautionary tale},
  author={Campbell, Peter and Groves, Michael and Shepherd, Dan},
  year={2014},
  howpublished={\href{http://docbox.etsi.org/Workshop/2014/201410_CRYPTO/S07_Systems_and_Attacks/S07_Groves_Annex.pdf}{\nolinkurl{docbox.etsi.org/Workshop/2014/201410_CRYPTO/S07_Systems_and_Attacks/S07_Groves_Annex.pdf}}},
  note={Presented at the ETSI 2nd Quantum-Safe Crypto Workshop}
}

@inproceedings{cramer,
  title={Recovering short generators of principal ideals in cyclotomic rings},
  author={Cramer, Ronald and Ducas, L{\'e}o and Peikert, Chris and Regev, Oded},
  booktitle={Annual International Conference on the Theory and Applications of Cryptographic Techniques},
  pages={559--585},
  year={2016},
  organization={Springer},
  doi={10.1007/978-3-662-49890-3_21}
}

@inproceedings{bauch,
  title={Short generators without quantum computers: {T}he case of multiquadratics},
  author={Bauch, Jens and Bernstein, Daniel J. and Valence, Henry de and Lange, Tanja and Vredendaal, Christine van},
  booktitle={Annual International Conference on the Theory and Applications of Cryptographic Techniques},
  pages={27--59},
  year={2017},
  organization={Springer},
  doi={10.1007/978-3-319-56620-7_2}
}

@article{pellet,
  author  = {Pellet-Mary, Alice and Hanrot, Guillaume and Stehl{\'e}, Damien},
  title   = {{Approx-SVP} in ideal lattices with pre-processing},
  journal = {Lecture Notes in Computer Science},
  year    = {2019},
  pages   = {685--716},
  doi     = {10.1007/978-3-030-17656-3_24}
}

@techreport{vanEmde,
  author      = {van Emde Boas, Peter},
  title       = {Another {NP}-complete problem and the complexity of computing short vectors in a lattice},
  institution = {Mathematical Institute, University of Amsterdam},
  year        = {1981},
  number      = {81-04},
  type        = {Report}
}

@article{dinur,
  author  = {Dinur, Irit and Kindler, Guy and Raz, Ran and Safra, Shmuel},
  title   = {Approximating {CVP} to Within Almost-Polynomial Factors is {NP}-Hard},
  journal = {Combinatorica},
  volume  = {23},
  number  = {2},
  pages   = {205--243},
  year    = {2003},
  publisher = {Springer}
}

@inproceedings{ajtai2,
  author    = {Ajtai, Mikl\'{o}s},
  title     = {The shortest vector problem in {$L_2$} is {NP}-hard for randomized reductions},
  booktitle = {Proceedings of the Thirtieth Annual ACM Symposium on Theory of Computing},
  year      = {1998},
  pages     = {10--19},
  publisher = {ACM},
  doi       = {10.1145/276698.276705}
}

@article{haviv,
  author  = {Haviv, Ishay and Regev, Oded},
  title   = {Tensor-based Hardness of the Shortest Vector Problem to within Almost Polynomial Factors},
  journal = {Theory of Computing},
  volume  = {8},
  number  = {1},
  pages   = {23--31},
  year    = {2012},
  doi     = {10.4086/toc.2012.v008a001}
}

@misc{wan,
    title={{NP}-hardness of {SVP} in {Euclidean} space}, 
    author={Wan, Daqing},
    howpublished = "\href{https://arxiv.org/abs/2603.27398}{\url{arXiv:2603.27398}}",
    year={2026}
}

@article{shoup,
  title={New algorithms for finding irreducible polynomials over finite fields},
  author={Shoup, Victor},
  journal={Mathematics of Computation},
  volume={54},
  number={189},
  pages={435--447},
  year={1990},
  publisher={American Mathematical Society}
}

@book{dickson,
  title={Linear Groups: With an Exposition of the Galois Field Theory},
  author={Dickson, Leonard Eugene},
  year={1901},
  publisher={B. G. Teubner},
  address={Leipzig}
}

@phdthesis{storjohann,
  title={Algorithms for matrix canonical forms},
  author={Storjohann, Arne},
  year={2000},
  school={ETH Zurich}
}

@inproceedings{rosca,
  title={{On the ring-LWE and polyno\-mial-LWE problems}},
  author={Rosca, Miruna and Stehl{\'e}, Damien and Wallet, Alexandre},
  booktitle={Annual International Conference on the Theory and Applications of Cryptographic Techniques},
  pages={146--173},
  year={2018},
  organization={Springer}
}

@inproceedings{rosen,
  title={Efficient collision-resistant hashing from worst-case assumptions on cyclic lattices},
  author={Peikert, Chris and Rosen, Alon},
  booktitle={Theory of Cryptography Conference},
  pages={145--166},
  year={2006},
  organization={Springer}
}

@inproceedings{peik,
  title={Pseudorandomness of ring-LWE for any ring and modulus},
  author={Peikert, Chris and Regev, Oded and Stephens-Davidowitz, Noah},
  booktitle={Proceedings of the 49th Annual ACM SIGACT Symposium on Theory of Computing},
  pages={461--473},
  year={2017}
}

@inproceedings{stehle,
  title={Efficient public key encryption based on ideal lattices},
  author={Stehl{\'e}, Damien and Steinfeld, Ron and Tanaka, Keisuke and Xagawa, Keita},
  booktitle={International Conference on the Theory and Application of Cryptology and Information Security},
  pages={617--635},
  year={2009},
  organization={Springer}
}

@inproceedings{lyub2,
  title={On ideal lattices and learning with errors over rings},
  author={Lyubashevsky, Vadim and Peikert, Chris and Regev, Oded},
  booktitle={Annual International Conference on the Theory and Applications of Cryptographic Techniques},
  pages={1--23},
  year={2010},
  organization={Springer}
}

@inproceedings{bennett,
  title={{Hardness of the (Approximate) Shortest Vector Problem: A Simple Proof via Reed-Solomon Codes}},
  author={Huck Bennett and Chris Peikert},
  booktitle={International Workshop on Approximation, Randomization, and Combinatorial Optimization: Algorithms and Techniques},
  year={2022},
  url={https://api.semanticscholar.org/CorpusID:246867325}
}

@misc{wang,
    title={Square-free discriminants of matrices and the generalized spectral characterizations of graphs}, 
    author={Wang, Wei and Yu, Tao},
    howpublished = "\href{https://arxiv.org/abs/1608.01144}{\url{arXiv:1608.01144}}",
    year={2016}
}

@inproceedings{bolb,
  title={{Order-LWE and the hardness of ring-LWE with entropic secrets}},
  author={Bolboceanu, Madalina and Brakerski, Zvika and Perlman, Renen and Sharma, Devika},
  booktitle={International Conference on the Theory and Application of Cryptology and Information Security},
  pages={91--120},
  year={2019},
  organization={Springer}
}

@article{pepin,
  title={Algebraically Structured LWE, Revisited},
  author={Peikert, Chris and Pepin, Zachary},
  journal={Journal of Cryptology},
  volume={37},
  number={3},
  pages={28},
  year={2024},
  publisher={Springer}
}

@article{bennett2,
  title={The complexity of the shortest vector problem},
  author={Bennett, Huck},
  journal={ACM SIGACT News},
  volume={54},
  number={1},
  pages={37--61},
  year={2023},
  publisher={ACM New York, NY, USA}
}

@article{decade,
  title={A decade of lattice cryptography},
  author={Peikert, Chris},
  journal={Foundations and Trends in Theoretical Computer Science},
  volume={10},
  number={4},
  pages={283--424},
  year={2016},
  publisher={Now Publishers, Inc.}
}

@book{goldwasser,
  title={Complexity of lattice problems: a cryptographic perspective},
  author={Micciancio, Daniele and Goldwasser, Shafi},
  volume={671},
  year={2012},
  publisher={Springer Science \& Business Media}
}

@misc{hair,
    title={{SVP$_p$ is NP-hard for all $p>2$, Even to approximate Within a factor of $\smash{2^{\log^{1-\varepsilon}n}}$}}, 
    author={Hair, Isaac and Sahai, Amit},
    howpublished = "\href{https://arxiv.org/html/2511.04125v1}{\url{arXiv:2511.04125}}",
    year={2026}
}

@book{basu,
  title     = {Algorithms in Real Algebraic Geometry},
  author    = {Basu, Saugata and Pollack, Richard and Roy, Marie-Fran{\c{c}}oise},
  year      = {2006},
  publisher = {Springer},
  address   = {Berlin, Heidelberg},
  doi       = {10.1007/3-540-33099-2}
}

@misc{stack,
    author = {GreginGre (math.stackexchange.com/users/447764/gregingre)},
    title = {Existence of orthogonal base for finite {G}alois extension over characteristic 2},
    howpublished = {Mathematics Stack Exchange},
    year = {2020},
    note = "\href{https://math.stackexchange.com/questions/3534461/existence-of-orthogonal-base-for-finite-galois-extension-over-characteristic-2}{\url{math.stackexchange.com/questions/3534461}} (accessed: 2026-09-02)"
}

@article{eberhard,
  title={The characteristic polynomial of a random matrix},
  author={Eberhard, Sean},
  journal={Combinatorica},
  volume={42},
  number={4},
  pages={491--527},
  year={2022},
  publisher={Springer}
}

@article{ferber,
  title={Random symmetric matrices: rank distribution and irreducibility of the characteristic polynomial},
  author={Ferber, Asaf and Jain, Vishesh and Sah, Ashwin and Sawhney, Mehtaab},
  journal={Mathematical Proceedings of the Cambridge Philosophical Society},
  volume={174},
  number={2},
  pages={257--270},
  year={2023},
  publisher={Cambridge University Press}
}

@article{vander,
  title={Die Seltenheit der Gleichungen mit Affekt},
  author={van der Waerden, Bartel Leendert},
  journal={Mathematische Annalen},
  volume={109},
  number={1},
  pages={13--16},
  year={1934},
  publisher={Springer}
}

@article{bhargava,
  author  = {Bhargava, Manjul and Shankar, Arul and Wang, Xiaoheng},
  title   = {Squarefree values of polynomial discriminants I},
  journal = {Inventiones mathematicae},
  year    = {2022},
  volume  = {228},
  pages   = {1037--1073},
  doi     = {10.1007/s00222-022-01098-w}
}

@article{lenstra,
  title={Algorithms in algebraic number theory},
  author={Lenstra, Hendrik W},
  journal={Bulletin of the American Mathematical Society},
  volume={26},
  number={2},
  pages={211--244},
  year={1992}
}

@article{shor,
  author  = {Shor, Peter W.},
  title   = {Polynomial-time algorithms for prime factorization and discrete logarithms on a quantum computer},
  journal = {SIAM Journal on Computing},
  year    = {1997},
  volume  = {26},
  pages   = {1484--1509},
  doi     = {10.1137/s0097539795293172}
}

@article{dusart,
  title={Explicit estimates of some functions over primes},
  author={Dusart, Pierre},
  journal={The Ramanujan Journal},
  volume={45},
  number={1},
  pages={227--251},
  year={2018},
  publisher={Springer}
}

@article{rosser,
  author  = {Rosser, J. Barkley and Schoenfeld, Lowell},
  title   = {Approximate formulas for some functions of prime numbers},
  journal = {Illinois Journal of Mathematics},
  year    = {1962},
  volume  = {6},
  doi     = {10.1215/ijm/1255631807}
}

@article{mildly,
  title={Mildly short vectors in cyclotomic ideal lattices in quantum polynomial time},
  author={Cramer, Ronald and Ducas, L{\'e}o and Wesolowski, Benjamin},
  journal={Journal of the ACM (JACM)},
  volume={68},
  number={2},
  pages={1--26},
  year={2021},
  publisher={ACM New York, NY, USA}
}

@misc{bhargava2,
    title={The geometric sieve and the density of squarefree values of invariant polynomials}, 
    author={Bhargava, Manjul},
    howpublished = "\href{https://arxiv.org/abs/1402.0031}{\url{arXiv:1402.0031}}",
    year={2014}
}

@article{ekedahl,
  title={An infinite version of the Chinese remainder theorem},
  author={Ekedahl, Torsten},
  journal={Commentarii mathematici Universitatis Sancti Pauli= Rikkyo Daigaku sugaku zasshi},
  volume={40},
  number={1},
  pages={53--59},
  year={1991},
  publisher={dummy publisher}
}

@misc{song,
    title={{Hardness of Euclidean Closest Vector within $n^{1/2-\epsilon}$ and Binary
Nearest Codeword within $n^{1-\epsilon}$}}, 
    author={Song, Zhao},
    howpublished = "Cryptology ePrint Archive, \href{https://eprint.iacr.org/2026/1655.pdf}{\url{iacr:2026/1655}}",
    year={2026}
}

@misc{liuCVP,
    title={{Exact CVP is NP-complete for principal cyclotomic
ideals}}, 
    author={Liu, Jiaqi and Feng, Yansong and Pan, Yanbin},
    howpublished = "Cryptology ePrint Archive, \href{https://eprint.iacr.org/2026/1793.pdf}{\url{iacr:2026/1793}}",
    year={2026}
}

@misc{liuSVP,
    title={{SVP is NP-Hard for some rank-2 cyclotomic modules}}, 
    author={Liu, Jiaqi and Feng, Yansong and Pan, Yanbin},
    howpublished = "Cryptology ePrint Archive, \href{https://eprint.iacr.org/2026/1856.pdf}{\url{iacr:2026/1856}}",
    year={2026}
}

@inproceedings{cai,
  title={{Approximating the SVP to within a factor $(1+1/\text{dim}^{\epsilon})$ is NP-hard under randomized conditions}},
  author={Cai, Jin-Yi and Nerurkar, Ajay},
  booktitle={Proceedings of the Thirteenth Annual IEEE Conference on Computational Complexity},
  pages={46--55},
  year={1998},
  organization={IEEE}
}

@article{miccSVP,
  title={The shortest vector in a lattice is hard to approximate to within some constant},
  author={Micciancio, Daniele},
  journal={SIAM Journal on Computing},
  volume={30},
  number={6},
  pages={2008--2035},
  year={2001},
  publisher={SIAM}
}

@article{khot,
  title={Hardness of approximating the shortest vector problem in high $\ell_p$ norms},
  author={Khot, Subhash},
  journal={Journal of Computer and System Sciences},
  volume={72},
  number={2},
  pages={206--219},
  year={2006}
}

@article{khot2,
  title={Hardness of approximating the shortest vector problem in lattices},
  author={Khot, Subhash},
  journal={Journal of the ACM (JACM)},
  volume={52},
  number={5},
  pages={789--808},
  year={2005}
}

@article{dumer,
  title={Hardness of approximating the minimum distance of a linear code},
  author={Dumer, Ilya and Micciancio, Daniele and Sudan, Madhu},
  journal={IEEE Transactions on Information Theory},
  volume={49},
  number={1},
  pages={22--37},
  year={2003},
  publisher={IEEE}
}

@misc{openai,
  author       = {{OpenAI}},
  title        = {Ten Advances in Mathematics and Theoretical Computer Science},
  note         = {Technical report, released August 1, 2026; updated August 6, 2026. \url{https://cdn.openai.com/pdf/ten-proofs-oai.pdf}}
}

\end{document}